\RequirePackage{hyperref}

\documentclass[letterpaper, 10 pt, conference]{ieeeconf}

\usepackage[utf8x]{inputenc}

\usepackage{amsmath}
\usepackage{amssymb}
\usepackage{stmaryrd}
\usepackage{xcolor}
\usepackage{listings}
\usepackage{xspace}
\usepackage{paralist}
\usepackage{url}
\usepackage{cite}
\usepackage{multirow}
\usepackage{subcaption}
\usepackage{thmtools}
\usepackage{thm-restate}

\usepackage{tikz}
\usetikzlibrary{arrows, automata, patterns, positioning}
\usepackage{pgfplots}
\usepackage{pgfplotstable}

\usepackage{comment}  

\usepackage{koord}  

\newcommand{\utmtool}{SkyTrakx\xspace}

\begin{document}

\makeatletter
\DeclareRobustCommand*\textsubscript[1]{%
	\@textsubscript{\selectfont#1}}
\def\@textsubscript#1{%
	{\m@th\ensuremath{_{\mbox{\fontsize\sf@size\z@#1}}}}}
\makeatother


\newcommand{\authcomment}[1]{\textbf{[[#1]]}}
\newcommand{\sayan}[1]{\textcolor{blue}{[[#1]]}}
\newcommand{\chiao}[1]{\textcolor{red}{[[#1]]}}
\newcommand{\hussein}[1]{\textcolor{orange}{[[#1]]}}
\newcommand{\sembrack}[1]{[\![#1]\!]}

\newcommand{\nnt}{{\sf T}^{\geq 0}}             
\newcommand{\post}{{\sf T}^{>0}}                
\newcommand{\Variables}{{\sf V}}                

\newcommand{\num}[1]{\relax\ifmmode \mathbb #1\else $\mathbb #1$\fi}
\newcommand{\nnnum}[1]{\relax\ifmmode 
  {\mathbb #1}_{\geq 0} \else ${\mathbb #1}_{\geq 0}$
  \fi}
\newcommand{\npnum}[1]{\relax\ifmmode 
  {\mathbb #1}_{\leq 0} \else ${\mathbb #1}_{\leq 0}$
  \fi}
\newcommand{\pnum}[1]{\relax\ifmmode 
  {\mathbb #1}_{> 0} \else ${\mathbb #1}_{> 0}$
  \fi}
\newcommand{\nnum}[1]{\relax\ifmmode 
  {\mathbb #1}_{< 0} \else ${\mathbb #1}_{< 0}$
  \fi}
\newcommand{\plnum}[1]{\relax\ifmmode 
  {\mathbb #1}_{+} \else ${\mathbb #1}_{+}$
  \fi}
\newcommand{\nenum}[1]{\relax\ifmmode 
  {\mathbb #1}_{-} \else ${\mathbb #1}_{-}$
  \fi}

\newcommand{\reals}{{\num R}}                    
\newcommand{\nnreals}{{\nnnum R}}                    
\newcommand{\realsinfty}{{\num R} \cup \{\infty, -\infty\}}                    
\newcommand{\plreals}{{\plnum R}}                    
\newcommand{\naturals}{{\num N}}                      
\newcommand{\integers}{{\num Z}}                      
\newcommand{\rationals}{{\num Q}}                      
\newcommand{\nnrationals}{{\nnnum Q}}                   
\newcommand{\Time}{{\num T}}  
\newcommand{\bools}{{\num B}}  
\newcommand{\plintegers}{{\plnum Z}}                      

\newcommand{\extb}[1]{\relax\ifmmode {\sf ExtBeh}_{#1} \else ${\sf ExtBeh}_{#1}$\fi} 
\newcommand{\tdists}[1]{\relax\ifmmode {\sf Tdists}_{#1} \else ${\sf Tdists}_{#1}$\fi} 

\newcommand{\exec}[1]{\relax\ifmmode {\sf Execs}_{#1} \else ${\sf Exec}_{#1}$\fi} 
\newcommand{\execf}[1]{\relax\ifmmode {\sf Execs}^*_{#1} \else ${\sf Exec}^*_{#1}$\fi} 
\newcommand{\execi}[1]{\relax\ifmmode {\sf Execs}^\omega_{#1} \else ${\sf Exec}^\omega_{#1}$\fi} 

\newcommand{\ctrace}[1]{\relax\ifmmode {\sf Ctraces}_{#1} \else ${\sf Ctraces}_{#1}$\fi} 

\newcommand{\trace}[1]{\relax\ifmmode {\sf Traces}_{#1} \else ${\sf Traces}_{#1}$\fi} 
\newcommand{\tracef}[1]{\relax\ifmmode {\sf Traces}^*_{#1} \else ${\sf Traces}^*_{#1}$\fi} 
\newcommand{\tracei}[1]{\relax\ifmmode {\sf Traces}^\omega_{#1} \else ${\sf Traces}^\omega_{#1}$\fi} 
\newcommand{\traceof}[1]{\relax\ifmmode {\sf trace}_{#1} \else ${\sf trace}_{#1}$\fi} 
\newcommand{\traceinv}[1]{\relax\ifmmode {\sf TraceInv}_{#1} \else ${\sf TraceInv}_{#1}$\fi}

\newcommand{\frag}[1]{\relax\ifmmode {\sf Frags}_{#1} \else ${\sf Frags}_{#1}$\fi} 
\newcommand{\fragf}[1]{\relax\ifmmode {\sf Frags}^*_{#1} \else ${\sf Frags}^*_{#1}$\fi} 
\newcommand{\fragi}[1]{\relax\ifmmode {\sf Frags}^\omega_{#1} \else ${\sf Frags}^\omega_{#1}$\fi} 

\newcommand{\reach}[1]{\relax\ifmmode {\sf Reach}_{#1} \else ${\sf Reach}_{#1}$\fi} 
\newcommand{\reachtube}[1]{\relax\ifmmode {\sf ReachTube}_{#1} \else ${\sf ReachTube}_{#1}$\fi} 

\newcommand{\pair}[2]{\relax\ifmmode \langle #1, #2 \rangle \else $\langle #1, #2 \rangle$\fi} 

\newcommand{\TE}{\relax\ifmmode \mathit{Time} \else $\mathit{Time}$ \fi} 
\newcommand{\EQ}{\relax\ifmmode \mathit{Enq} \else $\mathit{Enq}$ \fi} 
\newcommand{\DQ}{\relax\ifmmode \mathit{Deq} \else $\mathit{DeqTime}$ \fi} 
\newcommand{\E}{\relax\ifmmode \mathsf{E} \else $\mathsf{E}$ \fi}

\newcommand{\loc}{\relax\ifmmode \mathit{loc} \else $\mathit{loc}$ \fi}
\newcommand{\abs}{\relax\ifmmode \mathit{abs} \else $\mathit{abs}$ \fi}

\newcommand{\execs}{{\exec{}}}
\newcommand{\traces}{{\trace{}}}
\newcommand{\fragss}{{\frag{}}}

\newcommand{\fexecs}{{\execf{}}}
\newcommand{\ftraces}{{\tracef{}}}
\newcommand{\ffragss}{{\fragf{}}}

\newcommand{\iexecs}{{\execi{}}}
\newcommand{\itraces}{{\tracei{}}}
\newcommand{\ifragss}{{\fragi{}}}

\newcommand{\fstate}{{\sf fstate}}  
\newcommand{\lstate}{{\sf lstate}}  
\newcommand{\ltime}{{\sf ltime}}  
\newcommand{\clk}{{\mathit clk}}  
\newcommand{\msgs}{{\mathit msghist}}  

\newcommand{\length}[1]{\relax\ifmmode  \mathit{length}(#1) \else $\mathit{length}(#1)$\fi} 
\newcommand{\hbclose}[1]{\relax\ifmmode \mathit{before}(#1) \else $\mathit{before}(#1)$\fi} 
\newcommand{\haclose}[1]{\relax\ifmmode \mathit{after}(#1) \else $\mathit{after}(#1)$\fi} 
\newcommand{\trover}[1]{\relax\ifmmode \mathit{reach}_{#1} \else $\mathit{reach}_{#1}$\fi} 
\newcommand{\obsint}[1]{\relax\ifmmode \mathit{obs}_#1 \else $\mathit{obs}_#1$\fi} 
\newcommand{\Post}[3]{\relax\ifmmode \mathit{Post}_{#1}(#2,#3) \else $\mathit{Post}_{#1}(#2,#3)$\fi} 
\newcommand{\Pre}[3]{\relax\ifmmode \mathit{Pre}_{#1}(#2,#3) \else $\mathit{Pre}_{#1}(#2,#3)$\fi}

\newcommand{\pow}[1]{{\bf P}(#1)} 
\newcommand{\inverse}[1]{#1^{-1}}
\newcommand{\range}[1]{\ms{range{(#1)}}}
\newcommand{\domain}[1]{{\it dom}(#1)}
\newcommand{\type}[1]{\ms{type{(#1)}}}
\newcommand{\dtype}[1]{\ms{dtype{(#1)}}} 
\newcommand{\restr}{\mathrel{\lceil}}
\newcommand{\restrrange}{\mathrel{\downarrow}}
\newcommand{\point}[1]{\wp(#1)}                 
\newcommand{\proj}[2]{#1\!\!\upharpoonright_{#2}}

\def\A{{\cal A}} 
\def\B{{\cal B}} 
\def\C{{\cal C}} 
\def\D{{\cal D}} 
\def\E{{\cal E}} 
\def\F{{\cal F}} 
\def\G{{\cal G}} 
\def\H{{\cal H}} 
\def\I{{\cal I}} 
\def\K{{\cal K}} 
\def\L{{\cal L}} 
\def\M{{\cal M}} 
\def\O{{\cal O}} 
\def\P{{\cal P}} 
\def\Q{{\cal Q}} 
\def\R{{\cal R}} 
\def\S{{\cal S}} 
\def\T{{\cal T}} 
\def\V{{\cal V}} 
\def\U{{\cal U}} 
\def\X{{\cal X}} 
\def\Y{{\cal Y}} 
\def\Z{{\cal Z}} 


\newcommand{\col}[1]{\relax\ifmmode \mathscr #1\else $\mathscr #1$\fi}

\def\statemodels{\col{S}}

\definecolor{HIOAcolor}{rgb}{0.776,0.22,0.07}
\newcommand{\HIOA}{\textcolor{HIOAcolor}{\tt HIOA\hspace{3pt}}}
\newcommand{\PVS}{\textcolor{HIOAcolor}{\tt PVS\hspace{3pt}}}
\newcommand{\PVSnogap}{\textcolor{HIOAcolor}{\tt PVS\hspace{1pt}}}
\newcommand{\HIOAbiggap}{\textcolor{HIOAcolor}{\tt HIOA\hspace{6pt}}}
\newcommand{\HIOAnogap}{\textcolor{HIOAcolor}{\tt HIOA}}
\newcommand{\anyrelation}{\lessgtr}

\newcommand{\SC}[2]{\relax\ifmmode {\tt Scount}(#1,#2) \else ${\tt Scount}(#1,#2)$\fi} 
\newcommand{\SCM}[2]{\relax\ifmmode {\tt Smin}(#1,#2) \else ${\tt Smin}(#1,#2)$\fi} 
\newcommand{\Aut}[1]{\relax\ifmmode {\tt Aut}(#1) \else ${\tt Aut}(#1)$\fi} 

\newcommand{\auto}[1]{{\operatorname{\mathsf{#1}}}}
\newcommand{\act}[1]{{\operatorname{\mathsf{#1}}}}
\newcommand{\smodel}[1]{{\operatorname{\mathsf{#1}}}}
\newcommand{\pvstheory}[1]{{\operatorname{\mathit{#1}}}}

\newcommand{\Automaton}{{\bf automaton}}
\newcommand{\Asserts}{{\bf asserts}}
\newcommand{\Assumes}{{\bf assumes}}
\newcommand{\Backward}{{\bf backward}}
\newcommand{\By}{{\bf by}}
\newcommand{\Case}{{\bf case}}
\newcommand{\Choose}{{\bf  choose}}
\newcommand{\Components}{{\bf components}}
\newcommand{\Const}{{\bf const}}
\newcommand{\Converts}{{\bf converts}}
\newcommand{\Do}{{\bf do}}
\newcommand{\Eff}{{\bf eff}}
\newcommand{\Else}{{\bf else}}
\newcommand{\Elseif}{{\bf elseif}}
\newcommand{\Enumeration}{{\bf enumeration}}
\newcommand{\Ensuring}{{\bf ensuring}}
\newcommand{\Exempting}{{\bf exempting}}
\newcommand{\Fi}{{\bf fi}}
\newcommand{\For}{{\bf for}}
\newcommand{\Forward}{{\bf forward}}
\newcommand{\Freely}{{\bf freely}}
\newcommand{\From}{{\bf from}}
\newcommand{\Generated}{{\bf generated}}
\newcommand{\Local}{{\bf local}}
\newcommand{\Hidden}{{\bf hidden}}
\newcommand{\If}{{\bf if}}
\newcommand{\In}{{\bf in}}
\newcommand{\Implies}{{\bf implies}}
\newcommand{\Includes}{{\bf includes}}
\newcommand{\Introduces}{{\bf introduces}}
\newcommand{\Input}{{\bf input}}
\newcommand{\Kind}{{\bf kind}}
\newcommand{\Initially}{{\bf initially}}
\newcommand{\Internal}{{\bf internal}}
\newcommand{\Invariant}{{\bf invariant}}
\newcommand{\Od}{{\bf od}}
\newcommand{\Of}{{\bf of}}
\newcommand{\Output}{{\bf output}}
\newcommand{\Partitioned}{{\bf partitioned}}
\newcommand{\Signature}{{\bf signature}}
\newcommand{\Simulation}{{\bf simulation}}
\newcommand{\Sort}{{\bf sort}}
\newcommand{\States}{{\bf states}}
\newcommand{\Tasks}{{\bf tasks}}
\newcommand{\Then}{{\bf then}}
\newcommand{\To}{{\bf to}}
\newcommand{\Trait}{{\bf trait}}
\newcommand{\Traits}{{\bf traits}}
\newcommand{\Transitions}{{\bf transitions}}
\newcommand{\Tuple}{{\bf tuple}}
\newcommand{\Type}{{\bf type}}
\newcommand{\Union}{{\bf union}}
\newcommand{\Uses}{{\bf uses}}
\newcommand{\Where}{{\bf where}}
\newcommand{\While}{{\bf while}}
\newcommand{\With}{{\bf with}}

\newcommand{\FFF}{\vspace{0.1in}}
\newcommand{\BBB}{\hspace{-0.1in}}

\newcommand{\deq}{\mathrel{\stackrel{\scriptscriptstyle\Delta}{=}}}


\newcommand{\seclabel}[1]{\label{sec:#1}}
\newcommand{\secref}[1]{Section~\ref{sec:#1}}
\newcommand{\secreftwo}[2]{Sections~\ref{sec:#1}~and~\ref{sec:#2}}
\newcommand{\figlabel}[1]{\label{fig:#1}}
\newcommand{\figref}[1]{Figure~\ref{fig:#1}}
\newcommand{\figrefs}[2]{Figures~\ref{fig:#1} and~\ref{fig:#2}}
\newcommand{\applabel}[1]{\label{app:#1}}
\newcommand{\appref}[1]{Appendix~\ref{app:#1}}
\newcommand{\lnlabel}[1]{\label{line:#1}}
\newcommand{\lnrngref}[2]{lines~\ref{line:#1}--\ref{line:#2}\xspace}
\newcommand{\lnref}[1]{line~\ref{line:#1}\xspace}
\newcommand{\thmref}[1]{Theorem~\ref{thm:#1}\xspace}


\newcommand{\remove}[1]{}
\newcommand{\salg}[1]{\relax\ifmmode {\mathcal F}_{#1}\else ${\mathcal F}_{#1}$\fi} 
\newcommand{\msp}[1]{\relax\ifmmode (#1, \salg{#1}) \else $(#1, \salg{#1})$\fi} 
\newcommand{\msprod}[2]{\relax\ifmmode ( #1 \times #2, \salg{#1} \otimes \salg{#2}) \else $(#1 \times #2, \salg{#1} \otimes \salg{#2})$\fi} 
\newcommand{\dist}[1]{\relax\ifmmode {\mathcal P}\msp{#1}
  \else ${\mathcal P}\msp{#1}$\fi} 
\newcommand{\subdist}[1]{\relax\ifmmode {\mathcal S}{\mathcal P}\msp{#1} 
  \else ${\mathcal S}{\mathcal P}\msp{#1}$\fi} 
\newcommand{\disc}[1]{\relax\ifmmode {\sf Disc}(#1)
  \else ${\sf Disc}(#1)$\fi} 

\newcommand{\Trajeq}{\relax\ifmmode {\mathcal R}_\T \else ${\mathcal R}_\T$\fi} 
\newcommand{\Acteq}{\relax\ifmmode {\mathcal R}_A \else ${\mathcal R}_A$\fi} 
\newcommand{\noop}{\relax\ifmmode \lambda \else $\lambda$\fi} 
\newcommand{\close}[1]{\relax\ifmmode \overline{#1} \else $\overline{#1}$\fi} 

\newcommand{\corrtasks}{\mathop{\mathsf {c}}}
\newcommand{\full}{\mathop{\mathsf {full}}}
\newcommand{\tdist}{\mathop{\mathsf {tdist}}}
\newcommand{\extbeh}{\mathop{\mathsf {extbeh}}}
\newcommand{\apply}[2]{\mathop{\mathsf {apply}({#1},{#2})}}
\newcommand{\support}{\mathop{\mathsf {supp}}}
\newcommand{\maxrng}{\mathop{\mathsf {max}}}
\newcommand{\relation}{\mathrel{R}}
\newcommand{\cone}{C}
\newcommand{\flatten}{\mathord{\mathsf {flatten}}}
\newcommand{\discrete}{\mathord{\mathsf {Disc}}}
\newcommand{\lift}[1]{\mathrel{{\mathcal L}(#1)}}
\newcommand{\expansion}[1]{\mathrel{{\mathcal E}(#1)}}


\newcommand{\subdisc}{\operatorname{\mathsf {SubDisc}}}
\newcommand{\tran}{\operatorname{\mathsf {tran}}}

\renewcommand{\execs}{{\operatorname{\mathsf {Execs}}}}
\newcommand{\frags}{{\operatorname{\mathsf {Frags}}}}

\newcommand{\tracefnc}{{\operatorname{\mathsf {trace}}}}

\newcommand{\finite}{{\operatorname{\mathsf {finite}}}}
\newcommand{\hide}{{\operatorname{\mathsf {hide}}}}

\newcommand{\early}{{\operatorname{\mathsf {Early}}}}
\newcommand{\late}{{\operatorname{\mathsf {Late}}}}
\newcommand{\toss}{{\operatorname{\mathsf {Toss}}}}

\newcommand{\define}{:=}

\newcommand{\pc}{{\operatorname{\mathsf {counter}}}}
\newcommand{\chosen}{{\operatorname{\mathsf {chosen}}}}

\newcommand{\rand}{{\operatorname{\mathsf {random}}}}
\newcommand{\unif}{{\operatorname{\mathsf {unif}}}}

\newcommand{\ie}{i.e.,\xspace}
\newcommand{\Ie}{I.e.,\xspace}

\newcommand{\eg}{e.g.,\xspace}
\newcommand{\Eg}{E.g.,\xspace}


\newcommand{\mybox}[3]{
  \framebox[#1][l]
  {
    \parbox{#2}
    {
      #3
    }
  }
}

\newcommand{\two}[4]{
  \parbox{.98\columnwidth}{\vspace{1pt} \vfill
    \parbox[t]{#1\columnwidth}{#3}%
    \hspace{.5cm}
    \parbox[t]{#2\columnwidth}{#4}%
  }}

\newcommand{\twosep}[4]{
  \parbox{\columnwidth}{\vspace{1pt} \vfill
    \parbox[t]{#1\columnwidth}{#3}%
        \vrule width 0.2pt
    \parbox[t]{#2\columnwidth}{#4}%
  }}

\newcommand{\eqntwo}[4]{
  \parbox{\columnwidth}{\vspace{1pt} \vfill
    \parbox[t]{#1\columnwidth}{$ #3 $}
    \parbox[t]{#2\columnwidth}{$ #4 $}
  }}

\newcommand{\three}[6]{\vspace{1pt} \vfill
        \parbox{\columnwidth}{%
        \parbox[t]{#1\columnwidth}{#4}%
        \parbox[t]{#2\columnwidth}{#5}%
        \parbox[t]{#3\columnwidth}{#6}%
      }}

\newcommand{\tup}[1]
           {
             \relax\ifmmode
             \langle #1 \rangle
             \else $\langle$ #1 $\rangle$ \fi
           }

\newcommand{\lit}[1]{ \relax\ifmmode
                \mathord{\mathcode`\-="702D\sf #1\mathcode`\-="2200}
                \else {\it #1} \fi }

\newcommand{\figuresize}{\scriptsize}

\newcommand{\equationsize}{\footnotesize}

\lstdefinelanguage{ioa}{
  basicstyle=\figuresize,
  keywordstyle=\bf \figuresize,
  identifierstyle=\it \figuresize,
  emphstyle=\tt \figuresize,
  mathescape=true,
  tabsize=20,
  sensitive=false,
  columns=fullflexible,
  keepspaces=false,
  flexiblecolumns=true,
  basewidth=0.05em,
  moredelim=[il][\rm]{//},
  moredelim=[is][\sf \figuresize]{!}{!},
  moredelim=[is][\bf \figuresize]{*}{*},
  keywords={automaton,and, 
         choose,const,continue, components,
         discrete, do, derived,
         eff, external,else, elseif, evolve, end,each,
         fi,for, forward, from,
         hidden,
         in,input,internal,if,invariant, initially, imports,
     let,
     or, output, operators, od, of,
     pre, prob,
     return,
     such,satisfies, stop, signature, simulation, state, stochastic,
     trajectories,trajdef, transitions, that,then, type, types, to, tasks,
     variables, vocabulary, uni,
     when,where, with,while},
  emph={set, seq, tuple, map, array, enumeration},   
   literate=
        {(}{{$($}}1
        {)}{{$)$}}1
        {\\in}{{$\in\ $}}1
        {\\preceq}{{$\preceq\ $}}1
        {\\subset}{{$\subset\ $}}1
        {\\subseteq}{{$\subseteq\ $}}1
        {\\supset}{{$\supset\ $}}1
        {\\supseteq}{{$\supseteq\ $}}1
        {\\forall}{{$\forall$}}1
        {\\le}{{$\le\ $}}1
        {\\ge}{{$\ge\ $}}1
        {\\gets}{{$\gets\ $}}1
        {\\cup}{{$\cup\ $}}1
        {\\cap}{{$\cap\ $}}1
        {\\langle}{{$\langle$}}1
        {\\rangle}{{$\rangle$}}1
        {\\exists}{{$\exists\ $}}1
        {\\bot}{{$\bot$}}1
        {\\rip}{{$\rip$}}1
        {\\emptyset}{{$\emptyset$}}1
        {\\notin}{{$\notin\ $}}1
        {\\not\\exists}{{$\not\exists\ $}}1
        {\\ne}{{$\ne\ $}}1
        {\\to}{{$\to\ $}}1
        {\\implies}{{$\implies\ $}}1
        {<}{{$<\ $}}1
        {>}{{$>\ $}}1
        {=}{{$=\ $}}1
        {~}{{$\neg\ $}}1
        {|}{{$\mid$}}1
        {'}{{$^\prime$}}1
        {\\A}{{$\forall\ $}}1
        {\\E}{{$\exists\ $}}1
        {\\/}{{$\vee\,$}}1
        {\\vee}{{$\vee\,$}}1
        {/\\}{{$\wedge\,$}}1
        {\\wedge}{{$\wedge\,$}}1
        {=>}{{$\Rightarrow\ $}}1
        {->}{{$\rightarrow\ $}}1
        {<=}{{$\Leftarrow\ $}}1
        {<-}{{$\leftarrow\ $}}1
        {~=}{{$\neq\ $}}1
        {\\U}{{$\cup\ $}}1
        {\\I}{{$\cap\ $}}1
        {|-}{{$\vdash\ $}}1
        {-|}{{$\dashv\ $}}1
        {<<}{{$\ll\ $}}2
        {>>}{{$\gg\ $}}2
        {||}{{$\|$}}1
        {[}{{$[$}}1
        {]}{{$\,]$}}1
        {[[}{{$\langle$}}1
        {]]]}{{$]\rangle$}}1
        {]]}{{$\rangle$}}1
        {<=>}{{$\Leftrightarrow\ $}}2
        {<->}{{$\leftrightarrow\ $}}2
        {(+)}{{$\oplus\ $}}1
        {(-)}{{$\ominus\ $}}1
        {_i}{{$_{i}$}}1
        {_j}{{$_{j}$}}1
        {_{i,j}}{{$_{i,j}$}}3
        {_{j,i}}{{$_{j,i}$}}3
        {_0}{{$_0$}}1
        {_1}{{$_1$}}1
        {_2}{{$_2$}}1
        {_n}{{$_n$}}1
        {_p}{{$_p$}}1
        {_k}{{$_n$}}1
        {-}{{$\ms{-}$}}1
        {@}{{}}0
        {\\delta}{{$\delta$}}1
        {\\R}{{$\R$}}1
        {\\Rplus}{{$\Rplus$}}1
        {\\N}{{$\N$}}1
        {\\times}{{$\times\ $}}1
        {\\tau}{{$\tau$}}1
        {\\alpha}{{$\alpha$}}1
        {\\beta}{{$\beta$}}1
        {\\gamma}{{$\gamma$}}1
        {\\ell}{{$\ell\ $}}1
        {--}{{$-\ $}}1
        {\\TT}{{\hspace{1.5em}}}3        
      }

\lstdefinelanguage{ioaNums}[]{ioa}
{
  numbers=left,
  numberstyle=\tiny,
  stepnumber=2,
  numbersep=4pt
}

\lstdefinelanguage{ioaNumsRight}[]{ioa}
{
  numbers=right,
  numberstyle=\tiny,
  stepnumber=2,
  numbersep=4pt
}

\newcommand{\ioa}{\lstinline[language=IOA]}

\lstnewenvironment{IOA}%
  {\lstset{language=IOA}}
  {}

\lstnewenvironment{IOANums}%
  {
  \if@firstcolumn
    \lstset{language=IOA, numbers=left, firstnumber=auto}
  \else
    \lstset{language=IOA, numbers=right, firstnumber=auto}
  \fi
  }
  {}

\lstnewenvironment{IOANumsRight}%
  {
    \lstset{language=IOA, numbers=right, firstnumber=auto}
  }
  {}


\newcommand{\figioa}[5]{
  \begin{figure}[#1]
      \hrule \F
      {\figuresize \bf #2}
      \lstinputlisting[language=ioaLang]{#5}
      \F \hrule \F
      \caption{#3}
      \label{fig: #4}
  \end{figure}
}

\newcommand{\linefigioa}[9]{

}

\newcommand{\twofigioa}[8]{
  \begin{figure}[#1]
    \hrule \F
    {\figuresize \bf #2} \\
    \two{#5}{#6}
    {
      \lstinputlisting[language=ioaLang]{#7}
    }
    {
      \lstinputlisting[language=ioaLang]{#8}
    }
    \F \hrule \F
    \caption{#3}
    \label{fig: #4}
  \end{figure}
}

\lstdefinelanguage{ioaLang}{%
  basicstyle=\ttfamily\small,
  keywordstyle=\rmfamily\bfseries\small,
  identifierstyle=\small,
  keywords={assumes,automaton,axioms,backward,bounds,by,case,choose,components,const,d,det,discrete,do,eff,else,elseif,ensuring,enumeration,evolve,fi,fire,follow,for,forward,from,hidden,if,in,%
    input,initially,internal,invariant,let, local,od,of,output,pre,schedule,signature,so,%
    simulation,states,state,variables, tasks, stop,tasks,that,then,to,trajdef,trajectory,trajectories,transitions,tuple,type,
    uniform,union,urgent,uses,when,where,while,yield},
  literate=
        {\\in}{{$\in$}}1
        {\\preceq}{{$\preceq$}}1
        {\\subset}{{$\subset$}}1
        {\\subseteq}{{$\subseteq$}}1
        {\\supset}{{$\supset$}}1
        {\\supseteq}{{$\supseteq$}}1
        {\\rho}{{$\rho$}}1
        {\\infty}{{$\infty$}}1
        {<}{{$<$}}1
        {>}{{$>$}}1
        {=}{{$=$}}1
        {~}{{$\neg$}}1 
        {|}{{$\mid$}}1
        {'}{{$^\prime$}}1
        {\\A}{{$\forall$}}1 {\\E}{{$\exists$}}1
        {\\/}{{$\vee$}}1 {/\\}{{$\wedge$}}1 
        {=>}{{$\Rightarrow$}}1 
        {->}{{$\rightarrow$}}1 
        {<=}{{$\leq$}}1 {>=}{{$\geq$}}1 {~=}{{$\neq$}}1
        {\\U}{{$\cup$}}1 {\\I}{{$\cap$}}1
        {|-}{{$\vdash$}}1 {-|}{{$\dashv$}}1
        {<<}{{$\ll$}}2 {>>}{{$\gg$}}2
        {||}{{$\|$}}1
        {<=>}{{$\Leftrightarrow$}}2 
        {<->}{{$\leftrightarrow$}}2
        {(+)}{{$\oplus$}}1
        {(-)}{{$\ominus$}}1
}

\lstdefinelanguage{bigIOALang}{%
  basicstyle=\ttfamily,
  keywordstyle=\rmfamily\bfseries,
  identifierstyle=,
  keywords={assumes,automaton,axioms,backward,by,case,choose,components,const,%
    d,det,discrete,do,eff,else,elseif,ensuring,enumeration,evolve,fi,for,forward,from,hidden,if,in%
    input,initially,internal,invariant,local,od,of,output,pre,schedule,signature,so,%
    tasks, simulation,states,stop,tasks,that,then,to,trajdef,trajectories,transitions,tuple,type,union,urgent,uses,when,where,yield},
  literate=
        {\\in}{{$\in$}}1
        {\\preceq}{{$\preceq$}}1
        {\\subset}{{$\subset$}}1
        {\\subseteq}{{$\subseteq$}}1
        {\\supset}{{$\supset$}}1
        {\\supseteq}{{$\supseteq$}}1
        {<}{{$<$}}1
        {>}{{$>$}}1
        {=}{{$=$}}1
        {~}{{$\neg$}}1 
        {|}{{$\mid$}}1
        {'}{{$^\prime$}}1
        {\\A}{{$\forall$}}1 {\\E}{{$\exists$}}1
        {\\/}{{$\vee$}}1 {/\\}{{$\wedge$}}1 
        {=>}{{$\Rightarrow$}}1 
        {->}{{$\rightarrow$}}1 
        {<=}{{$\leq$}}1 {>=}{{$\geq$}}1 {~=}{{$\neq$}}1
        {\\U}{{$\cup$}}1 {\\I}{{$\cap$}}1
        {|-}{{$\vdash$}}1 {-|}{{$\dashv$}}1
        {<<}{{$\ll$}}2 {>>}{{$\gg$}}2
        {||}{{$\|$}}1
        {<=>}{{$\Leftrightarrow$}}2 
        {<->}{{$\leftrightarrow$}}2
        {(+)}{{$\oplus$}}1
        {(-)}{{$\ominus$}}1
}

\lstnewenvironment{BigIOA}%
  {\lstset{language=bigIOALang,basicstyle=\ttfamily}
   \csname lst@SetFirstLabel\endcsname}
  {\csname lst@SaveFirstLabel\endcsname\vspace{-4pt}\noindent}

\lstnewenvironment{SmallIOA}%
  {\lstset{language=ioaLang,basicstyle=\ttfamily\scriptsize}
   \csname lst@SetFirstLabel\endcsname}
  {\csname lst@SaveFirstLabel\endcsname\noindent}

\newcommand{\true}{\relax\ifmmode \mathit true \else \em true \/\fi}
\newcommand{\false}{\relax\ifmmode \mathit false \else \em false \/\fi}

\newcommand{\Real}{{\operatorname{\texttt{Real}}}}
\newcommand{\Bool}{{\operatorname{\texttt{Bool}}}}
\newcommand{\Char}{{\operatorname{\texttt{Char}}}}
\newcommand{\ioaInt}{{\operatorname{\texttt{Int}}}}
\newcommand{\ioaNat}{{\operatorname{\texttt{Nat}}}}
\newcommand{\ioaAugR}{{\operatorname{\texttt{AugmentedReal}}}}
\newcommand{\ioaString}{{\operatorname{\texttt{String}}}}
\newcommand{\Discrete}{{\operatorname{\texttt{Discrete}}}}
\newcommand{\limplies}{\Rightarrow}
\newcommand{\liff}{\Leftrightarrow}

\newlength{\bracklen}
\newcommand{\sem}[1]{\settowidth{\bracklen}{[}
     [\hspace{-0.5\bracklen}[#1]\hspace{-0.5\bracklen}]}

\newcommand{\defaultArraystretch}{1.4}
\renewcommand{\arraystretch}{\defaultArraystretch}

\newcommand{\gS}{\mathcal{S}}
\newcommand{\gV}{\mathcal{V}}
\newcommand{\freevars}{\mathcal{FV}}

\newcommand{\gVspec}{\mathcal{V}_\mathit{spec}}
\newcommand{\gVa}{\mathcal{V}_\mathit{A}}
\newcommand{\gVsig}{\mathcal{V}_\mathit{sigs}}
\newcommand{\gVso}{\mathcal{V}_\mathit{sorts}}
\newcommand{\gVop}{\mathcal{V}_\mathit{ops}}
\newcommand{\sort}{\mathit{sort}}
\newcommand{\sig}{\mathit{sig}}
\newcommand{\id}{\mathit{id}}
\newcommand{\sigsep}{\lsl`->`}

\newcommand{\super}[2]{\ensuremath{\mathit{#1}^\mathit{#2}}}
\newcommand{\tri}[3]{\ensuremath{\mathit{#1}^\mathit{#2}_\mathit{#3}}}
\newcommand{\Assumptions}{\ensuremath{\mathit{Assumptions}}}
\newcommand{\actPred}[3][\pi]{\tri{P}{#2,#1}{#3}}
\newcommand{\actualTypes}[1]{\super{actualTypes}{#1}}
\newcommand{\actuals}[1]{\super{actuals}{#1}}
\newcommand{\autActVars}[2][\pi]{\vars{#2}{},\vars{#2,#1}{}}
\newcommand{\bracket}[2]{\mathit{#1}[\mathit{#2}]}
\newcommand{\compVars}[1]{\super{compVars}{#1}}
\newcommand{\context}{\mathit{context}}
\newcommand{\ensuring}[2]{\tri{ensuring}{#1}{#2}}
\newcommand{\initPred}[1]{\tri{P}{#1}{init}}
\newcommand{\initVals}[1]{\super{initVals}{#1}}
\newcommand{\initially}[2]{\tri{initially}{#1}{#2}}
\newcommand{\invPred}[2]{\tri{Inv}{#1}{#2}}
\newcommand{\knownVars}[1]{\super{knownVars}{#1}}
\newcommand{\localPostVars}[2]{\tri{localPostVars}{#1}{#2}}
\newcommand{\localVars}[2]{\tri{localVars}{#1}{#2}}
\newcommand{\locals}[1]{\bracket{Locals}{#1}}
\newcommand{\nam}[1]{\rho^{\mathit{#1}}}
\newcommand{\otherActPred}[3][\pi]{\otherTri{P}{#2,#1}{#3}}
\newcommand{\otherParams}[2]{\otherTri{params}{#1}{#2}}
\newcommand{\otherSub}[2]{\otherTri{\sigma}{#1}{#2}}
\newcommand{\otherTri}[3]{\tri{\smash{#1'}}{#2}{#3}}
\newcommand{\otherVars}[2]{\otherTri{vars}{#1}{#2}}
\newcommand{\params}[2]{\tri{params}{#1}{#2}}
\newcommand{\postVars}[1]{\super{postVars}{#1}}
\newcommand{\pre}[2]{\tri{Pre}{#1}{#2}}
\newcommand{\prog}[2]{\tri{Prog}{#1}{#2}}
\newcommand{\prov}[2]{\tri{Prov}{#1}{#2}}
\newcommand{\stateSorts}[1]{\super{stateSorts}{#1}}
\newcommand{\stateVars}[1]{\super{stateVars}{#1}}
\newcommand{\states}[1]{\bracket{States}{#1}}
\newcommand{\sub}[2]{\tri{\sigma}{#1}{#2}}
\newcommand{\sugActPred}[3][\pi]{\tri{P}{#2,#1}{#3,desug}}
\newcommand{\sugLocalVars}[2]{\ifthenelse{\equal{}{#2}}%
                             {\tri{localVars}{#1}{desug}}%
                             {\tri{localVars}{#1}{#2,desug}}}
\newcommand{\sugVars}[2]{\ifthenelse{\equal{}{#2}}%
                        {\tri{vars}{#1}{desug}}%
                        {\tri{vars}{#1}{#2,desug}}}
\newcommand{\cVars}[1]{\super{cVars}{#1}}

\newcommand{\vmap}{\dot{\varrho}}
\newcommand{\map}[2]{\tri{\vmap}{#1}{#2}}

\newcommand{\types}[1]{\super{types}{#1}}
\newcommand{\vars}[2]{\tri{vars}{#1}{#2}}

\newcommand{\subActPred}[3][\pi]{\sub{#2,#1}{#3}(\tri{P}{#2,#1}{#3,desug})}
\newcommand{\subLocalVars}[2]{\sub{#1}{#2}(\tri{localVars}{#1}{#2,desug})}

\newcommand{\dA}{\hat{A}}
\newcommand{\renameAction}[1]{\ensuremath{\rho_{#1}(\vars{\dA{#1},\pi}{})}}
\newcommand{\renameComponent}[1]{\ensuremath{\rho_{#1}\dA_{#1}}}

\newenvironment{Syntax}{\[\begin{subSyntax}}{\end{subSyntax}\]\vspace{-.3in}}
\newenvironment{subSyntax}{\begin{array}{l}}{\end{array}}
\newcommand{\w}[1]{\mbox{\hspace*{#1em}}}

\newcommand{\ms}[1]{\ifmmode%
\mathord{\mathcode`-="702D\it #1\mathcode`\-="2200}\else%
$\mathord{\mathcode`-="702D\it #1\mathcode`\-="2200}$\fi}


\newcommand{\kw}[1]{{\bf #1}} 
\newcommand{\tcon}[1]{{\tt #1}} 
\newcommand{\syn}[1]{{\tt #1}} 
\newcommand{\pvskw}[1]{{\sc #1}} 
\newcommand{\pvsid}[1]{{\operatorname{\mathit{#1}}}}

\def\A{{\cal A}} 
\def\B{{\cal B}} 
\def\D{{\cal D}} 
\def\T{{\cal T}} 

\newcommand{\vu}{{\bf u}}
\newcommand{\vv}{{\bf v}}
\newcommand{\vw}{{\bf w}}
\newcommand{\vx}{{\bf x}}
\newcommand{\vy}{{\bf y}}
\newcommand{\va}{{\bf a}}
\newcommand{\vb}{{\bf b}}
\newcommand{\vq}{{\bf q}}
\newcommand{\vs}{{\bf s}}
\newcommand{\vm}{{\bf m}}
\newcommand{\vp}{{\bf p}}

\newcommand{\arrow}[1]{\mathrel{\stackrel{#1}{\rightarrow}}}
\newcommand{\sarrow}[2]{\mathrel{\stackrel{#1}{\rightarrow_{#2}}}}
\newcommand{\concat}{\mathbin{^{\frown}}} 
\newcommand{\paste}{\mathrel{\diamond}}
\newcommand{\Arrow}[1]{\mathrel{\stackrel{#1}{\Longrightarrow}}}


\newcommand{\carlane}[1]{\mathit{lane({#1})}}
\newcommand{\carnext}[1]{\mathit{next({#1})}}
\newcommand{\carprev}[1]{\mathit{prev({#1})}}
\newcommand{\carleft}[1]{\mathit{left({#1})}}
\newcommand{\carright}[1]{\mathit{right({#1})}}
\newcommand{\carego}{e}

\newcommand{\relative}[1]{\mathit{rel}(#1)}

\newcommand{\und}{{\tt U}} 

\def\CC{{\mathscr C}} 


\lstdefinelanguage{pvs}{
  basicstyle=\tt \figuresize,
  keywordstyle=\sc \figuresize,
  identifierstyle=\it \figuresize,
  emphstyle=\tt \figuresize,
  mathescape=true,
  tabsize=20,
  sensitive=false,
  columns=fullflexible,
  keepspaces=false,
  flexiblecolumns=true,
  basewidth=0.05em,
  moredelim=[il][\rm]{//},
  moredelim=[is][\sf \figuresize]{!}{!},
  moredelim=[is][\bf \figuresize]{*}{*},
  keywords={and, 
         begin,
         cases, const,
         do,
         external, else, exists, end, endcases, endif,
         fi,for, forall, from,
         hidden,
         in, if, importing,
     let, lambda, lemma,
     measure, 
     not,
     or, of,
     return, recursive,
     stop, 
     theory, that,then, type, types, type+, to, theorem,
     var,
     with,while},
  emph={nat, setof, sequence, eq, tuple, map, array, enumeration, bool, real, exp, nnreal, posreal},   
   literate=
        {(}{{$($}}1
        {)}{{$)$}}1
        {\\in}{{$\in\ $}}1
        {\\mapsto}{{$\rightarrow\ $}}1
        {\\preceq}{{$\preceq\ $}}1
        {\\subset}{{$\subset\ $}}1
        {\\subseteq}{{$\subseteq\ $}}1
        {\\supset}{{$\supset\ $}}1
        {\\supseteq}{{$\supseteq\ $}}1
        {\\forall}{{$\forall$}}1
        {\\le}{{$\le\ $}}1
        {\\ge}{{$\ge\ $}}1
        {\\gets}{{$\gets\ $}}1
        {\\cup}{{$\cup\ $}}1
        {\\cap}{{$\cap\ $}}1
        {\\langle}{{$\langle$}}1
        {\\rangle}{{$\rangle$}}1
        {\\exists}{{$\exists\ $}}1
        {\\bot}{{$\bot$}}1
        {\\rip}{{$\rip$}}1
        {\\emptyset}{{$\emptyset$}}1
        {\\notin}{{$\notin\ $}}1
        {\\not\\exists}{{$\not\exists\ $}}1
        {\\ne}{{$\ne\ $}}1
        {\\to}{{$\to\ $}}1
        {\\implies}{{$\implies\ $}}1
        {<}{{$<\ $}}1
        {>}{{$>\ $}}1
        {=}{{$=\ $}}1
        {~}{{$\neg\ $}}1
        {|}{{$\mid$}}1
        {'}{{$^\prime$}}1
        {\\A}{{$\forall\ $}}1
        {\\E}{{$\exists\ $}}1
        {\\/}{{$\vee\,$}}1
        {\\vee}{{$\vee\,$}}1
        {/\\}{{$\wedge\,$}}1
        {\\wedge}{{$\wedge\,$}}1
        {->}{{$\rightarrow\ $}}1
        {=>}{{$\Rightarrow\ $}}1
        {->}{{$\rightarrow\ $}}1
        {<=}{{$\Leftarrow\ $}}1
        {<-}{{$\leftarrow\ $}}1
        {~=}{{$\neq\ $}}1
        {\\U}{{$\cup\ $}}1
        {\\I}{{$\cap\ $}}1
        {|-}{{$\vdash\ $}}1
        {-|}{{$\dashv\ $}}1
        {<<}{{$\ll\ $}}2
        {>>}{{$\gg\ $}}2
        {||}{{$\|$}}1
        {[}{{$[$}}1
        {]}{{$\,]$}}1
        {[[}{{$\langle$}}1
        {]]]}{{$]\rangle$}}1
        {]]}{{$\rangle$}}1
        {<=>}{{$\Leftrightarrow\ $}}2
        {<->}{{$\leftrightarrow\ $}}2
        {(+)}{{$\oplus\ $}}1
        {(-)}{{$\ominus\ $}}1
        {_i}{{$_{i}$}}1
        {_j}{{$_{j}$}}1
        {_{i,j}}{{$_{i,j}$}}3
        {_{j,i}}{{$_{j,i}$}}3
        {_0}{{$_0$}}1
        {_1}{{$_1$}}1
        {_2}{{$_2$}}1
        {_n}{{$_n$}}1
        {_p}{{$_p$}}1
        {_k}{{$_n$}}1
        {-}{{$\ms{-}$}}1
        {@}{{}}0
        {\\delta}{{$\delta$}}1
        {\\R}{{$\R$}}1
        {\\Rplus}{{$\Rplus$}}1
        {\\N}{{$\N$}}1
        {\\times}{{$\times\ $}}1
        {\\tau}{{$\tau$}}1
        {\\alpha}{{$\alpha$}}1
        {\\beta}{{$\beta$}}1
        {\\gamma}{{$\gamma$}}1
        {\\ell}{{$\ell\ $}}1
        {--}{{$-\ $}}1
        {\\TT}{{\hspace{1.5em}}}3        
      }

\lstdefinelanguage{BigPVS}{
  basicstyle=\tt,
  keywordstyle=\sc,
  identifierstyle=\it,
  emphstyle=\tt ,
  mathescape=true,
  tabsize=20,
  sensitive=false,
  columns=fullflexible,
  keepspaces=false,
  flexiblecolumns=true,
  basewidth=0.05em,
  moredelim=[il][\rm]{//},
  moredelim=[is][\sf \figuresize]{!}{!},
  moredelim=[is][\bf \figuresize]{*}{*},
  keywords={and, 
         begin,
         cases, const,
         do, datatype,
         external, else, exists, end, endif, endcases,
         fi,for, forall, from,
         hidden,
         in, if, importing,
     let, lambda, lemma,
     measure,
     not,
     or, of,
     return, recursive,
     stop, 
     theory, that,then, type, types, type+, to, theorem,
     var,
     with,while},
  emph={nat, setof, sequence, eq, tuple, map, array, first, rest, add, enumeration, bool, real, posreal, nnreal},   
   literate=
        {(}{{$($}}1
        {)}{{$)$}}1
        {\\in}{{$\in\ $}}1
        {\\mapsto}{{$\rightarrow\ $}}1
        {\\preceq}{{$\preceq\ $}}1
        {\\subset}{{$\subset\ $}}1
        {\\subseteq}{{$\subseteq\ $}}1
        {\\supset}{{$\supset\ $}}1
        {\\supseteq}{{$\supseteq\ $}}1
        {\\forall}{{$\forall$}}1
        {\\le}{{$\le\ $}}1
        {\\ge}{{$\ge\ $}}1
        {\\gets}{{$\gets\ $}}1
        {\\cup}{{$\cup\ $}}1
        {\\cap}{{$\cap\ $}}1
        {\\langle}{{$\langle$}}1
        {\\rangle}{{$\rangle$}}1
        {\\exists}{{$\exists\ $}}1
        {\\bot}{{$\bot$}}1
        {\\rip}{{$\rip$}}1
        {\\emptyset}{{$\emptyset$}}1
        {\\notin}{{$\notin\ $}}1
        {\\not\\exists}{{$\not\exists\ $}}1
        {\\ne}{{$\ne\ $}}1
        {\\to}{{$\to\ $}}1
        {\\implies}{{$\implies\ $}}1
        {<}{{$<\ $}}1
        {>}{{$>\ $}}1
        {=}{{$=\ $}}1
        {~}{{$\neg\ $}}1
        {|}{{$\mid$}}1
        {'}{{$^\prime$}}1
        {\\A}{{$\forall\ $}}1
        {\\E}{{$\exists\ $}}1
        {\\/}{{$\vee\,$}}1
        {\\vee}{{$\vee\,$}}1
        {/\\}{{$\wedge\,$}}1
        {\\wedge}{{$\wedge\,$}}1
        {->}{{$\rightarrow\ $}}1
        {=>}{{$\Rightarrow\ $}}1
        {->}{{$\rightarrow\ $}}1
        {<=}{{$\Leftarrow\ $}}1
        {<-}{{$\leftarrow\ $}}1
        {~=}{{$\neq\ $}}1
        {\\U}{{$\cup\ $}}1
        {\\I}{{$\cap\ $}}1
        {|-}{{$\vdash\ $}}1
        {-|}{{$\dashv\ $}}1
        {<<}{{$\ll\ $}}2
        {>>}{{$\gg\ $}}2
        {||}{{$\|$}}1
        {[}{{$[$}}1
        {]}{{$\,]$}}1
        {[[}{{$\langle$}}1
        {]]]}{{$]\rangle$}}1
        {]]}{{$\rangle$}}1
        {<=>}{{$\Leftrightarrow\ $}}2
        {<->}{{$\leftrightarrow\ $}}2
        {(+)}{{$\oplus\ $}}1
        {(-)}{{$\ominus\ $}}1
        {_i}{{$_{i}$}}1
        {_j}{{$_{j}$}}1
        {_{i,j}}{{$_{i,j}$}}3
        {_{j,i}}{{$_{j,i}$}}3
        {_0}{{$_0$}}1
        {_1}{{$_1$}}1
        {_2}{{$_2$}}1
        {_n}{{$_n$}}1
        {_p}{{$_p$}}1
        {_k}{{$_n$}}1
        {-}{{$\ms{-}$}}1
        {@}{{}}0
        {\\delta}{{$\delta$}}1
        {\\R}{{$\R$}}1
        {\\Rplus}{{$\Rplus$}}1
        {\\N}{{$\N$}}1
        {\\times}{{$\times\ $}}1
        {\\tau}{{$\tau$}}1
        {\\alpha}{{$\alpha$}}1
        {\\beta}{{$\beta$}}1
        {\\gamma}{{$\gamma$}}1
        {\\ell}{{$\ell\ $}}1
        {--}{{$-\ $}}1
        {\\TT}{{\hspace{1.5em}}}3        
      }

\lstdefinelanguage{pvsNums}[]{pvs}
{
  numbers=left,
  numberstyle=\tiny,
  stepnumber=2,
  numbersep=4pt
}

\lstdefinelanguage{pvsNumsRight}[]{pvs}
{
  numbers=right,
  numberstyle=\tiny,
  stepnumber=2,
  numbersep=4pt
}

\newcommand{\pvs}{\lstinline[language=PVS]}

\lstnewenvironment{BigPVS}%
  {\lstset{language=BigPVS}}
  {}

\lstnewenvironment{PVSNums}%
  {
  \if@firstcolumn
    \lstset{language=pvs, numbers=left, firstnumber=auto}
  \else
    \lstset{language=pvs, numbers=right, firstnumber=auto}
  \fi
  }
  {}

\lstnewenvironment{PVSNumsRight}%
  {
    \lstset{language=pvs, numbers=right, firstnumber=auto}
  }
  {}

\newcommand{\figpvs}[5]{
  \begin{figure}[#1]
      \hrule \F
      {\figuresize \bf #2}
      \lstinputlisting[language=pvs]{#5}
      \F \hrule \F
      \caption{#3}
      \label{fig: #4}
  \end{figure}
}

\newcommand{\linefigpvs}[9]{

}

\newcommand{\twofigpvs}[8]{
  \begin{figure}[#1]
    \hrule \F
    {\figuresize \bf #2} \\
    \two{#5}{#6}
    {
      \lstinputlisting[language=pvsLang]{#7}
    }
    {
      \lstinputlisting[language=pvsLang]{#8}
    }
    \F \hrule \F
    \caption{#3}
    \label{fig: #4}
  \end{figure}
}

\lstdefinelanguage{pvsproof}{
  basicstyle=\tt \figuresize,
  mathescape=true,
  tabsize=4,
  sensitive=false,
  columns=fullflexible,
  keepspaces=false,
  flexiblecolumns=true,
  basewidth=0.05em,
}


\newcommand{\contract}{OV\xspace}
\newcommand{\Contracts}{OVs\xspace}
\newcommand{\contracts}{{\bf OV}\xspace}
\newcommand{\ID}{\ensuremath{\mathit{ID}}}
\newcommand{\asmap}{\ensuremath{\mathit{map}}}

\newtheorem{assumption}{Assumption}{\bfseries}{\itshape}
\newtheorem{definition}{Definition}{\bfseries}{\itshape}
\newtheorem{proposition}{Proposition}{\bfseries}{\itshape}
\newtheorem{lemma}{Lemma}{\bfseries}{\itshape}

\newcommand{\propref}[1]{Proposition~\ref{prop:#1}\xspace}
\newcommand{\proplabel}[1]{\label{prop:#1}}
\newcommand{\lemref}[1]{Lemma~\ref{lem:#1}\xspace}
\newcommand{\lemlabel}[1]{\label{lem:#1}}

\newcommand{\tabref}[1]{Table~\ref{table:#1}}
\newcommand{\tablabel}[1]{\label{table:#1}}

\newcommand{\AirMgr}{\textit{AM}\xspace}
\newcommand{\agent}{\textit{agent}\xspace}
\newcommand{\Sys}{\textit{Sys}\xspace}
\newcommand{\shift}{\mathit{shift}}
\newcommand{\protocol}{\mathtt{protocol}}
\newcommand{\bound}{\mathit{bound}}
\newcommand{\dur}{\mathit{dur}}
\newcommand{\dest}{\mathit{dest}}
\newcommand{\len}{\mathit{len}}
\newcommand{\last}{\mathit{last}}
\newcommand{\all}{\mathit{all}}
\newcommand{\AccAgents}{\mathit{AccAgents}}

\newcommand{\opervol}{OV\xspace}
\newcommand{\opervols}{OVs\xspace}

\newcommand{\Conse}{\textsc{Conservative}\xspace}
\newcommand{\Aggre}{\textsc{Aggressive}\xspace}

\setlength{\abovecaptionskip}{0pt}
\setlength{\belowcaptionskip}{0pt}
\setlength{\abovedisplayskip}{3pt}
\setlength{\belowdisplayskip}{3pt}
\setlength{\textfloatsep}{7pt}

\title{\utmtool: A Toolkit for Simulation and Verification of Unmanned Air-Traffic Management Systems (Extended Version)}

\author{%
    \authorblockN{Chiao Hsieh,
        Hussein Sibai,
        Hebron Taylor,
        Yifeng Ni,
        and
        Sayan Mitra
    }\\
    \authorblockA{
        University of Illinois at Urbana-Champaign\\
        Email: \{chsieh16,sibai2,hdt2,yifengn2,mitras\}@illinois.edu
    }
}

\maketitle
\thispagestyle{empty}
\pagestyle{empty}

\begin{abstract}
The key concept for safe and efficient traffic management for Unmanned Aircraft Systems (UAS) is the notion of \emph{operation volume} (OV).
An OV is a 4-dimensional block of airspace and time, which  can express an aircraft's \emph{intent},
and can be used for planning, de-confliction, and traffic management.
While there are several high-level simulators for UAS Traffic Management (UTM),
we are lacking a framework for creating, manipulating, and reasoning about OVs for heterogeneous air vehicles.
In this paper, we address this  and present  \utmtool---a software toolkit for simulation and verification of UTM scenarios based on \opervols.
First, we illustrate a use case of \utmtool by presenting a specific air traffic coordination protocol.
This protocol communicates OVs between  participating aircraft and an airspace manager for traffic routing.
We show how existing formal verification tools, \emph{Dafny} and \emph{Dione}, can assist in automatically checking key properties of the protocol.
Second, we show how the OVs can be computed for heterogeneous air vehicles like quadcopters and fixed-wing aircraft using another verification technique, namely \emph{reachability analysis}.
Finally, we show that \utmtool can be used to simulate complex scenarios involving heterogeneous vehicles,
for testing and performance evaluation in terms of workload and response delays analysis.
Our experiments delineate the trade-off between performance and workload across different strategies for generating \opervols.
\end{abstract}

\section{Introduction}
\label{sec:intro}

\emph{Unmanned Aircraft Traffic Management (UTM)} is an ecosystem of technologies that aim to enable unmanned, autonomous and human-operated, air vehicles to be used for transportation, delivery, and surveillance.
By 2024, 1.48 million recreational and 828 thousand  commercial unmanned aircraft are expected to be flying in the US national airspace~\cite{faa_forecast}.
Unlike the commercial airspace, this emerging area will have to accommodate heterogeneous and innovative vehicles relying on real-time distributed coordination, federated enforcement of regulations, and lightweight training for safety.
NASA, FAA, and a number of corporations are vigorously developing various UTM concepts, use cases, information architectures, and protocols towards the envisioned future where a large number of autonomous air vehicles can safely operate beyond visual line-of-sight.

FAA's UTM ConOps~\cite{utm_conops} defines the basic principles for safe coordination in UTM and the roles and responsibilities for the different parties involved such as the vehicle operator, manufacturer, the airspace service provider, and the FAA.
The building-block concept in UTM is the notion of \emph{operation volumes (\opervols)} which are used to share \emph{intent information} that, in turn,
enables interactive planning and strategic de-confliction for multiple UAS~\cite{utm_conops}.
Roughly, \opervols are 4D blocks of airspace with time intervals. They are used to specify the space that UAS is allowed to occupy over an interval of time (see Figures~\ref{fig:vehicles} and~\ref{fig:citysim}).
While there have been small-scale field tests for UTM protocols using \opervols~\cite{upp_summary},
there remains a strong need for a general-purpose framework for simulating and verifying UTM protocols based on \opervols.
Such a framework will need to
\begin{inparaenum}[(i)]
\item manipulate and communicate \opervols for traffic management protocols,
\item reason about dynamic \opervols for establishing safety of the protocols,
\item compute \opervols for heterogeneous air vehicles performing different maneuvers, and
\item evaluate UTM protocols in different simulation environments.
\end{inparaenum}

In this paper, we address this need and present \emph{\utmtool---an open source toolkit for simulation and verification of UTM scenarios}.
The toolkit offers a framework that
\begin{inparaenum}[(i)]
\item provides automata theory-based APIs for designing UTM protocols that formalize the communication of \opervols,
\item integrates existing tools, Dafny and Dione, to assist in verifying the safety and liveness of the protocols,
\item uses the reachability analysis tool DryVR to compute \opervols for heterogeneous air vehicles, and
\item expands the ROS and Gazebo-based CyPhyHouse framework~\cite{cyphyhouse_icra2020} to simulate and evaluate configurable UTM scenarios.
\end{inparaenum}
Benefit from~\cite{cyphyhouse_icra2020}, the protocols can be ported from simulations to hardware implementations.
The detailed contributions of \utmtool are as follows:

\paragraph*{Provably safe De-conflicting using \opervols}
For the first time we show how the intention expressed as \opervols can ensure provably safe distributed de-conflicting in Sections~\ref{sec:contract} and \ref{sec:protocol}.
As an example, we develop an automata-based de-conflicting protocol using \utmtool APIs.
This protocol specifies how the participating agents, the air vehicles, should interact with the Airspace Manager (\AirMgr).
We then formally verify the safety and liveness of this protocol.
In general, verification of distributed algorithms is challenging, but
our safety analysis shows that the use of \opervols helps decompose the global de-conflicting of the UAS into local invariant on the \AirMgr and local real-time requirements on each agent.
We further show that Dione~\cite{dione2019}, a proof assistant for Input/Output Automata (IOA) built with the Dafny program analyzer~\cite{leino_dafny:_2010}, can prove the local invariant on the \AirMgr automatically.
We prove that the safety of the protocol is achieved when individual agents follow their declared \opervols.
The liveness analysis further shows that every agent can eventually find a non-conflicting \opervol, under a stricter set of assumptions.

\paragraph*{Reachability Analysis for \opervol Conformance}
The guarantees of our protocol relies on the assumption that the agents do not violate their declared \opervols.
In~\secref{heuristics}, we show how to use an existing data-driven reachability analysis tool, DryVR~\cite{FanQMVCAV017}, to
create \opervols for heterogeneous air vehicles with low violation probability.
We apply such analysis on a quadrotor model, Hector Quadrotor~\cite{hector_quadrotor},
and a fixed-wing aircraft model, ROSplane~\cite{ROSplane}, and incorporate them in \utmtool.
We show both air vehicles in Figure~\ref{fig:vehicles} and visualize their \opervols for a landing scenario in Figure~\ref{fig:citysim}.

\paragraph*{Performance Evaluation}
In~\secref{experiment}, we first discuss the implementation of  \utmtool. Then, we
perform a detailed empirical analysis of our protocol in a number of representative scenarios using \utmtool. We compare two strategies for the generation of \opervols with different aggressiveness, namely \Conse and \Aggre.
Our experiments quantify the performance and workload on the \AirMgr,
and we measure these metrics with respect to the number of participating agents and different strategies for generating \opervols.
Our results suggest that the workload on the \AirMgr scales linearly with the number of agents,
and \Aggre provides 1.5-3X speedup but leads to 2-5X increased workload on the \AirMgr.

\begin{figure}[t!]
    \centering
    \includegraphics[width=0.49\columnwidth]{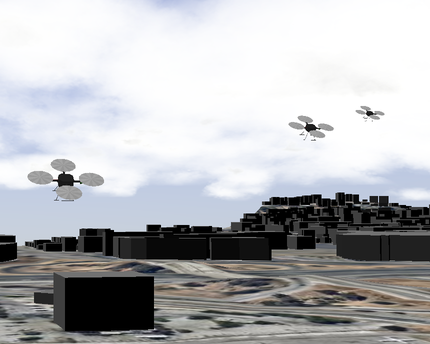}
    \hfill
    \includegraphics[width=0.49\columnwidth]{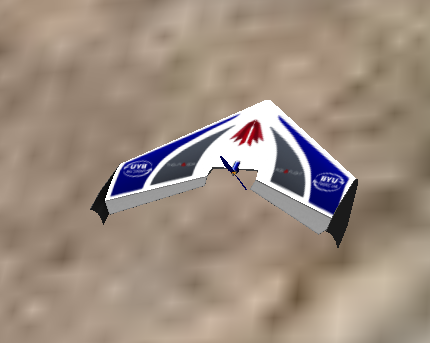}

    \caption{\small Hector Quadrotor~\cite{hector_quadrotor}~(\emph{Left}) and ROSplane~\cite{ROSplane}~(\emph{Right}) models in Gazebo simulator.}\label{fig:vehicles}
\end{figure}

\begin{figure}[t!]
    \begin{subfigure}[t]{0.49\columnwidth}
        \includegraphics[width=\columnwidth]{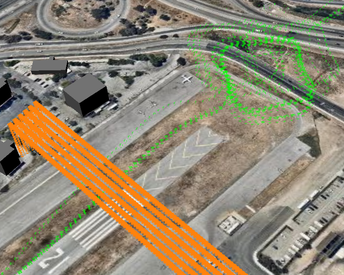}
        \caption{ROSplane reserved \opervols for loitering and descending.}
    \end{subfigure}
    \hfill
    \begin{subfigure}[t]{0.49\columnwidth}
        \includegraphics[width=\columnwidth]{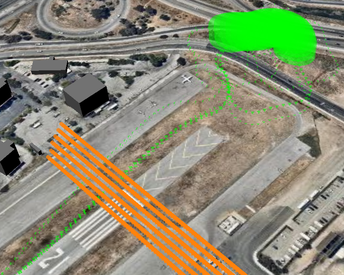}
        \caption{ROSplane loiters and waits for Quadrotors.}
    \end{subfigure}

    \vspace{3pt}
    \begin{subfigure}[t]{0.49\columnwidth}
        \includegraphics[width=\columnwidth]{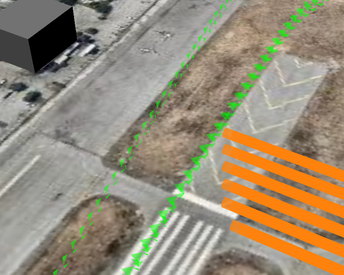}
        \caption{Quadrotors passed the runway before ROSplane descends.}
    \end{subfigure}
    \hfill
    \begin{subfigure}[t]{0.49\columnwidth}
        \includegraphics[width=\columnwidth]{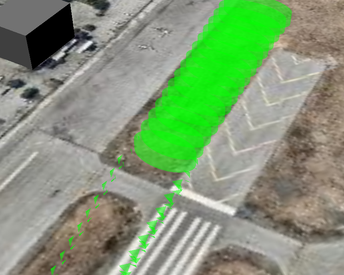}
        \caption{ROSplane descends.}
    \end{subfigure}
    \caption{\small Visualization of a landing scenario with heterogeneous air vehicles in an airport.
             The \opervols for Hector Quadrotors are annotated with orange and \opervols for the ROSplane are shown in green.
             Reserved \opervols are outlined with dots, and \opervols in use are represented with solid tubes.
    }\label{fig:citysim}
\end{figure}

\section{Related Work}
\label{sec:related}


\paragraph*{Collision Avoidance Protocols}

Prior to the development of the UTM ecosystem, traffic management protocols for manned aircraft include the family of Traffic Alert and Collision Avoidance Systems (TCAS)~\cite{tcas_ii,kochenderfer2012next,ACAS,ACASXu,air_safety_protocols_survey,Unmanned_Safety_survey}.
UTM and TCAS are complementary---the former is for long range strategic safety against loss of separation with other aircraft and  static obstacles, weather events, and anomalous behaviors, while the latter is for shorter-range tactical safety.
Accordingly the protocol we discuss (in~\secref{protocol}) coordinates over longer range and not \emph{only} for potential collision avoidance. \utmtool could be augmented with existing collision avoidance protocols in the future.   For instance, if an aircraft violates its \opervol in our protocol,
then a TCAS-like protocol can be used to avoid collision.





\paragraph*{Formal Approaches to UTM and Collision Avoidance}

%
The formal methods' research community has engaged with the problem of  air-traffic management in a number of different ways.
There have been several works on  formal analysis of TCAS~\cite{TCAS_verification_1997,livadasRTSS99,AirTrafficManagementSystem1999,Livadas2000HighlevelMA},
ACAS X~\cite{ACASX_verif_tacas,ACASX_verif_emsoft,reluplex},
and other protocols~\cite{JM:2012:small,ALAS2014,TaylorDistributedVerification,formalmethodsDrone,ModelBasedVerif_SemiAutonomous,UmenoL07,ALAS2014}.\footnote{\scriptsize\url{https://ti.arc.nasa.gov/news/acasx-verification-software/}}
These verification efforts rely on various simplifying assumptions such as precise state estimates,
straight-line trajectories, constant velocity of the intruder and ownership.
%
Algorithms to synthesize safe-by-construction plans for multiple drones flying in a shared airspace have been developed in~\cite{flybylogic,drona,schouwenaars2006safe,cyphyhouse_icra2020}.
These approaches rely on predicting and communicating future behavior of participating aircraft under different sources of uncertainty \cite{TaylorDistributedVerification,flybylogic,drona}.
%

In~\cite{bharadwaj2019traffic}, the authors present an approach for decentralized policy synthesis for route planning of individual vehicles modeled as Markov decision processes.
Our approach decouples the low-level dynamically feasible planning from the distributed coordination,
and solves the latter problem using a centralized coordinator (Airspace Manager) via distributed mutual exclusion over regions of the airspace (Section~\ref{sec:protocol}).
In \cite{bharadwaj2021traffic_TCNS}, the authors present a framework for decentralized controller synthesis for different managers of neighboring airspaces.
They use finite game and assume-guarantee approaches to generate decision-making mechanisms that satisfy linear temporal logic specifications.
An application of their approach is to design policies for airspace managers that enforce a maximum number of vehicles in the airspace or maximum loitering time.
Their framework assumes the operating regions for actions such as takeoff or loitering are predefined.
Our framework is complementary to this work as we show how a vehicle can generate an \opervol based on its vehicle dynamics from infinite choices of regions and time.

\section{A Formal Model of Operation Volumes}
\label{sec:contract}

\newcommand{\ToSet}[1]{\ensuremath{\llbracket #1 \rrbracket}}

In this section, we formalize the notion of \opervols described in~\cite{utm_conops} which is the fundamental building block for UTM protocols.
This formalization is also implemented in \utmtool for creating, manipulating, and reasoning about \opervols.
We refer to a UAS participating in the UTM system as an \emph{agent}, or equivalently, an \emph{air vehicle}.
Every agent in the system has a unique identifier.
The set of all possible identifiers is $\ID$.
We assume that each agent has access to a common global clock  which takes non-negative real numbers.
The 
\emph{airspace} is modeled as a compact subset $\X \subseteq \reals^3$.
Large airspaces may have to be divided into several smaller airspaces,
and one has to deal with hand-off across airspaces.
In this paper, we do not handle this problem of air vehicles entering and leaving $\X$. Other works have synthesized safe protocols for this problem (e.g. \cite{bharadwaj2021traffic_TCNS}).
The airspace is different from the state space of individual air vehicles
which may have many other state components like velocity, acceleration, pitch and yaw angles, etc.
Informally, an \opervol is a schedule for an air vehicle for occupying airspace.
%
\begin{definition}
\label{def:contract}
An \emph{operating volume (\contract)}
is a finite sequence of pairs
$C = (R_1, T_1),(R_2,T_2)$, $\ldots,(R_k,T_k)$
where each $R_i \subseteq \X$ is a compact subset of the airspace,
and $T_i$'s is a monotonically increasing sequence of time points.
\end{definition}

The total \emph{time duration} $T_k - T_1$ of the \contract $C$ is denoted by $C.\dur$,
and the length $k$ of $C$ is denoted by $C.\len$.
Further, we denote the last time point $T_k$ by $C.T_\last$, the last region $R_k$ by $C.R_\last$,
and the union of all regions, $\bigcup_{i=1}^k R_i$, by $C.R_\all$ as shorthands.
We denote the set of all possible contracts as $\contracts$.
An air vehicle meets an \contract at real-time $t$ if
\begin{inparaenum}[(1)]
\item $t \in [T_i, T_{i+1})$  for any $i < k$ implies that the air vehicle is located within $R_i$, and
\item $t \geq T_k$ implies that the agent is located within $R_k$ \emph{ever after $T_k$}.
\end{inparaenum}

\begin{definition}
    Two \Contracts are \emph{time-aligned} if they use the same sequence of time points.
    Given two time-aligned \Contracts,
    $C^a = (R^a_1, T_1)$, $\ldots,(R^a_k,T_k)$ and
    $C^b = (R^b_1, T_1)$, $\ldots,(R^b_k,T_k)$,
    and a set operation $\oplus \in \{\cap, \cup, \setminus\}$, we define
    \[
    C^a \oplus C^b \triangleq (R^a_1 \oplus R^b_1, T_1),\dotsc,(R^a_k \oplus R^b_k, T_k).
    \]
\end{definition}
We can generalize the definition to \Contracts that are not \emph{time-aligned},
and the detail derivation is provided in Appendix~\ref{appx:contract}.

Several concepts are defined naturally set operations on \Contracts.
We abuse notation sometimes and use $C$ as the set represented by contract $C$, i.e. the set
\begin{align*}
C \triangleq&\bigcup_{i=1}^{k-1} \{(r, t) \mid r \in R_i \land T_i \leq t < T_{i+1}\} \\
            &\cup \{(r, t) \mid r \in R_k \land T_k \leq t\}.
\end{align*}
For example, checking if $C^a$ \emph{refines} $C^b$ is to simply check
if $C^a$ uses less space-time than $C^b$ does,
i.e., $C^a \subseteq C^b$, or equivalently $C^a \setminus C^b = \emptyset$.

We will use the defined operations in our protocol in Section~\ref{sec:protocol} to update \opervols of individual agents and check intersections.
We will show how to create such \opervols using reachability analysis in Section~\ref{sec:heuristics}.

\section{A Simple Coordination Protocol using \Contracts}
\label{sec:protocol}
We present an example protocol for safe traffic management using \Contracts and its correctness argument.
We further implement the protocol with \utmtool.
The protocol involves a set of agents communicating \Contracts with an \emph{airspace manager or controller (\AirMgr)}.
The overall system is the composition of the airspace manager~(\AirMgr) and all agents~($\agent_i$):
\[
\mathit{Sys} \triangleq \AirMgr || \{\agent_i\}_{i\in\ID}.
\]

In Section~\ref{subsec:protocol-manager} and~\ref{subsec:protocol-agent}, we describe the protocol by showing the interaction between
participating agents and the \AirMgr through \inlinekrd{request}, \inlinekrd{reply}, and \inlinekrd{release} messages.
We then analyze the safety of the protocol under instant message delivery in Section~\ref{subsec:protocol-analysis},
and its liveness in Section~\ref{subsec:protocol-liveness}.

\subsection{Airspace Manager}
\label{subsec:protocol-manager}

\begin{figure}[b!]
    \begin{center}
\begin{lstlisting}[language=NumKoord]
automaton AirspaceManager

variables:
    contr_arr: [$\ID$ -> $\contracts$]
    reply_set: Set$\langle\ID\rangle$

output reply_i(contr: $\contracts$ = contr_arr[$i$])  $\label{code:am-reply}$
    pre: $i \in $ reply_set
    eff: reply_set := reply_set $\setminus$ {$i$}

input request_i(contr: $\contracts$)  $\label{code:am-request}$
    eff:
        reply_set := reply_set $\cup$ {$i$}
        if $\bigwedge\limits^{j \in \ID}_{j\neq i}$ (contr $\cap$ contr_arr[$j$] = $\emptyset$): $\label{code:am-disjointness}$
            contr_arr[$i$] := contr_arr[$i$]$\cup$contr $\label{code:am-addcontract}$

input release_i(contr: $\contracts$)
    eff: contr_arr[i] := contr_arr[i] $\setminus$ contr $\label{code:am-release}$
\end{lstlisting}
        \caption{Airspace Manager automaton.
            A model in Dione language~\cite{dione2019} with automated invariant checking for IOA is available in Appendix~\ref{appx:airmgr-dione}.}
        \label{fig:manager}
    \end{center}
\end{figure}

We design the \AirMgr as an Input/Output Automaton~(IOA)~\cite{lynch1996a} defined in \figref{manager}.
The \AirMgr keeps track of all contracts and checks for conflicts before approving new contracts.
It uses a mapping \inlinekrd{contr_arr} in which \inlinekrd{contr_arr[i]} records the contract held by agent $i$,
and a set \inlinekrd{reply_set} to store the agents whose requests are being processed and pending reply.

Whenever the \AirMgr receives a \inlinekrd{request_i(contr)} from agent $i$ (line~\ref{code:am-request}),
agent $i$ is first added to \inlinekrd{reply_set}.
Then, \inlinekrd{contr} is checked against all contracts of other agents by checking disjointness (line~\ref{code:am-disjointness}).
Only if the check succeeds, \inlinekrd{contr} is included in \inlinekrd{contr_arr[i]} via set union (line~\ref{code:am-addcontract}).

When $i$ is in \inlinekrd{reply_set},
the \inlinekrd{reply_i(contr)} action is triggered to reply to agent $i$ with the recorded \inlinekrd{contr=contr_arr[i]}~(line~\ref{code:am-reply}).
Note that the \AirMgr replies with the \emph{recorded contract} \inlinekrd{contr_arr[i]} at line~\ref{code:am-reply} irrespective of whether the \emph{requested contract} \inlinekrd{contr} in line~\ref{code:am-request} was included in \inlinekrd{contr_arr[i]} or not.
Finally, if the \AirMgr receives a \inlinekrd{release_i(contr)},
then it removes \inlinekrd{contr} from \inlinekrd{contr_arr[i]} via set difference (line~\ref{code:am-release}).

\subsection{Agent Protocol}
\label{subsec:protocol-agent}

The agent's coordination protocol sits in between a \emph{planner/navigator} that proposes \Contracts
and a \emph{controller} which drives the air vehicle to its target.
We will discuss approaches to estimate \Contracts for waypoint based path planners and waypoint following controllers in \secref{heuristics}.
Figure~\ref{fig:stateflow} shows the simplified state diagram of the agent protocol.
At a high level, agent $i$'s protocol starts in the idle state and initiates when
a \inlinekrd{plan} action with a given \inlinekrd{contr} is triggered by the agent's planner.
Then, the protocol requests this contract from the \AirMgr, and waits for the reply.
If the requested contract is a subset of the one replied by the \AirMgr, the agent protocol enters the moving state.
At this point the agent's controller starts moving the air vehicle and ideally following the contract strictly.
Once the air vehicle reaches the last region of \contract successfully, the protocol releases the unnecessary portion of the contract and goes back to idle state.
In the case that the requested contract is not a subset of the one replied by the \AirMgr, the protocol directly releases and retries.
If the agent violates the contract while moving, it notifies the \AirMgr that the contract is violated.
We provide the formally specified automaton and detail explanation of agent's protocol in Appendix~\ref{appx:protocol}.

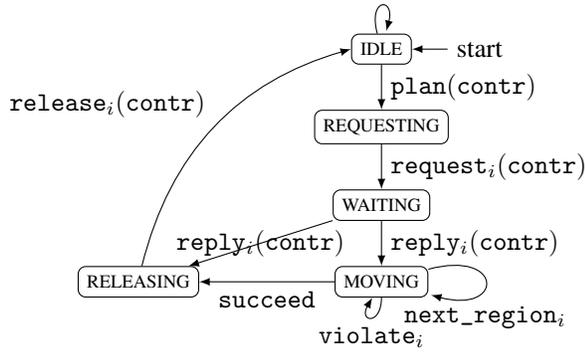
\begin{figure}[t!]
\centering
\begin{tikzpicture}[->, >=latex, scale=1,
    state/.style={draw,shape=rectangle,rounded corners=0.1cm,minimum width=0.7cm, node distance=0.6cm, font=\scriptsize}
]

\node[state, initial right]       (IDLE)       {IDLE};
\node[state, below=of IDLE]       (REQUESTING) {REQUESTING};
\node[state, below=of REQUESTING] (WAITING)    {WAITING};
\node[state, below=of WAITING]    (MOVING)     {MOVING};
\node[state, left=1.8cm of MOVING](RELEASING)  {RELEASING};

\path (IDLE)       edge node [right] {\inlinekrd{plan(contr)}}      (REQUESTING)
      (IDLE)       edge [loop above] (IDLE)
      (REQUESTING) edge node [right] {\inlinekrd{request_i(contr)}} (WAITING)
      (WAITING)    edge node [right] {\inlinekrd{reply_i(contr)}}   (MOVING)
      (MOVING)     edge [loop right] node [below=0.2cm] {\inlinekrd{next_region_i}}   (MOVING)
      (MOVING)     edge node [below] {\inlinekrd{succeed}}   (RELEASING)
      (MOVING)     edge [out=270,in=240, looseness=8] node [below] {\inlinekrd{violate_i}}   (MOVING)
      (RELEASING)  edge [bend left] node [above left] {\inlinekrd{release_i(contr)}} (IDLE.west)
      (WAITING)    edge node {\inlinekrd{reply_i(contr)}}   (RELEASING)
    ;

\end{tikzpicture}

\caption{Simplified state diagram for Agent.}\label{fig:stateflow}
\end{figure}

\subsection{Protocol Correctness: Safety}
\label{subsec:protocol-analysis}

We now discuss the safety property ensured by our protocol.
Here, $\agent_i.\texttt{curr\_contr}$ denotes the contract that the $i^{\mathit{th}}$ agent is following.
Assuming that none of the agents triggered their \inlinekrd{violate} action,
then an agent always follows its local contract \inlinekrd{curr_contr}.
In that case, collision avoidance is defined naturally as the disjointness between the \inlinekrd{curr_contr}s of all agents.
Our goal therefore is to show that the following proposition is an invariant of the system:
\begin{proposition}[Safety]\label{prop:safety}
If none of the agents triggered their $\texttt{violate}$ action, the current contracts followed by all agents are pairwise disjoint, i.e.,
\[
\small
\bigwedge_{i \in \ID} \bigwedge_{j\neq i,\\ j \in \ID} \agent_i.\texttt{curr\_contr} \cap \agent_j.\texttt{curr\_contr} = \emptyset.
\]
\end{proposition}

Our proof strategy is to show that first the global record of contracts maintained by the \AirMgr are pairwise disjoint by \lemref{am-contr}.
Then, we ensure the local copy by each agent is as restrictive as the global record and hence preserves disjointness by \lemref{agt-contr}.
With \lemref{am-contr} and \lemref{agt-contr}, \propref{safety} is derived following basic set theory.
We start from \lemref{am-contr} for the \AirMgr.
\begin{lemma}
	\lemlabel{am-contr}
If none of the agents triggered their \inlinekrd{violate} action, all contracts recorded by the \AirMgr are pairwise disjoint, i.e.,
\[
\bigwedge_{i \in \ID} \bigwedge_{j\neq i, j \in \ID} \AirMgr.\texttt{contr\_arr}[i] \cap \AirMgr.\texttt{contr\_arr}[j] = \emptyset.
\]
\end{lemma}
\begin{proof}
This is a direct result from examining all actions of the \AirMgr automaton.
The \inlinekrd{request_i} action ensures that a \inlinekrd{contr} is only included into \inlinekrd{contr_arr[i]}
if it is disjoint from all other contracts \inlinekrd{contr_arr[j]}.
The \inlinekrd{reply_i} action does not modify \inlinekrd{contr_arr} at all,
and \inlinekrd{release_i} action only shrinks the contracts.
\end{proof}
\begin{lemma}\lemlabel{agt-contr}
If none of the agents triggered their \inlinekrd{violate} action,
the \inlinekrd{curr_contr} of agent $i$ is always as restrictive as \inlinekrd{contr_arr[i]}, i.e.,
\[
\bigwedge_{i \in \ID} \agent_i.\texttt{curr\_contr} \subseteq \AirMgr.\texttt{contr\_arr}[i].
\]
\end{lemma}

\begin{proof}
This is proven by examining all actions of agent automaton regardless of the order of execution.
Due to the space limit, we only consider when actions are delivered instantaneously.
The \inlinekrd{curr_contr} is only modified in \inlinekrd{reply} and \inlinekrd{release} actions.
In \inlinekrd{reply} action, \inlinekrd{curr_contr} is to copy \inlinekrd{contr} sent by the \AirMgr and thus \lemref{agt-contr} holds.
In \inlinekrd{release} action, \inlinekrd{curr_contr} removes \inlinekrd{contr} first;
then \inlinekrd{release} is delivered to the \AirMgr to remove \inlinekrd{contr}.
As a result, \lemref{agt-contr} still holds.
In~Appendix~\ref{appx:safety-delayed}, we extend the proof so that, even under delayed communication settings, the lemma still holds when the order of received messages is preserved.
\end{proof}

\subsection{Protocol Correctness: Liveness}\label{subsec:protocol-liveness}

\newcommand{\reschedule}{\ensuremath{\mathit{reschedule}}\xspace}

For liveness property, we would like to see every agent eventually reaches its target.
In our protocol, this is formulated as every agent eventually reaches the last region of its \contract that it proposed in \inlinekrd{plan} action
and triggers its \inlinekrd{succeed} action.
The overall proof is to show that an agent can always find an \opervol which the \AirMgr approves.

Since a newly proposed \contract may be rejected,
we denote it as \inlinekrd{plan_contr} to distinguish from \inlinekrd{curr_contr} which an agent always follows.
It is worth noting that liveness depends on the \contract for each agent.
A simple scenario where liveness cannot be achieved is when the final destinations of two agents are too close;
thus the last region where one agent stays at the end could block the other agent forever.
Therefore, we first require the following assumption:
\begin{assumption}[Disjointness of different agents' regions]
	\label{ass:disjointness}
For any agent $i\in\ID$, all regions that it plans to traverse are disjoint from the last regions of all other agents. Formally,
	\[
        \bigwedge_{j\neq i} \texttt{plan\_contr}_i.R_\all \cap \AirMgr.\texttt{contr\_arr}[j].R_\last = \emptyset.
	\]
\end{assumption}
Assumption~\ref{ass:disjointness} can be achieved by querying the \AirMgr when planning
since \lemref{agt-contr} ensures the \AirMgr's record of \Contracts includes the agents' \Contracts.
\begin{definition}
Given an \contract $C=(R_1, T_1)$, $\ldots,(R_k,T_k)$ and a time duration $\delta$,
we define $\reschedule(C, \delta)$ as:
{\small
\[
\reschedule(C, \delta) \triangleq (R_1, T_1 + \delta), (R_2, T_2 + \delta), \ldots,(R_k, T_k + \delta)
\]
}
\end{definition}

Now we start our argument for liveness.
By our protocol design, if agent~$i$ never violates its \contract,
it must reach the last region successfully.
Therefore, we only have to prove that agent $i$'s request to the \AirMgr must be accepted eventually.
With Assumption~\ref{ass:disjointness},
we prove the claim that an agent $i$ can always \emph{reschedule} a plan so that the \AirMgr approves its \contract.

\begin{proposition}[Liveness]\label{prop:liveness}
If $\texttt{plan\_contr}_i$ satisfies Assumption~\ref{ass:disjointness},
then there is a time duration $\delta_0$ such that
the \AirMgr approves $\reschedule(\texttt{plan\_contr}_i, \delta)$ for all $\delta \geq \delta_0$.
Formally,
\begin{align*}
\bigwedge_{j\neq i, j \in \ID} \reschedule(\texttt{plan\_contr}_i, \delta)\ \cap \nonumber \\ \AirMgr.\texttt{contr\_arr}[j] = \emptyset.
\end{align*}
\end{proposition}
\begin{proof}
Following Assumption~\ref{ass:disjointness},
we first derive the disjointness of regions of airspace.
For any $j\neq i$ and any $\delta$,
\begin{align}\label{eq:R-disjoint}
\reschedule(\texttt{plan\_contr}_i, \delta).R_\all \nonumber \\
\cap\ \AirMgr.\texttt{contr\_arr}[j].R_\last = \emptyset
\end{align}
because \reschedule does not modify the regions.
Further, we derive that any $\delta_j \geq \AirMgr.\texttt{contr\_arr}[j].T_\last$,
the following two \Contracts are disjoint:
\begin{equation}\label{eq:OVC-disjoint}
\small
\reschedule(\texttt{plan\_contr}_i, \delta_j) \cap \AirMgr.\texttt{contr\_arr}[j] = \emptyset
\end{equation}
The proof is to expand the definition and skipped here.
Intuitively, this is because every agent $j$ are expected to reach and stay in $\AirMgr.\texttt{contr\_arr}[j].R_\last$ ever after
$\delta_j \geq \AirMgr.\texttt{contr\_arr}[j].T_\last$.
Therefore,
the rescheduled \contract for agent $i$ does not overlap with \Contracts of any other agent $j$.

Finally, let
\(
\delta_0 \triangleq \max\limits_{j\neq i}\ \AirMgr.\texttt{contr\_arr}[j].T_\last
\) and it directly leads to the proof of \propref{liveness}.
\end{proof}

In addition to the manual proof presented,
we have also explored using Dione~\cite{dione2019} with Dafny proof assistant~\cite{leino_dafny:_2010} to generate induction proof for invariants of IOA.
We choose this tool due to its support for IOA and automated SMT solving for set operations on \Contracts.
We discover that the tools can automatically prove the local invariant \lemref{am-contr} for the \AirMgr.
However, it lacks support for continuous time to model agents and communication delay;
hence we cannot use Dione prove other lemmas and propositions directly.

\section{Reachability Analysis and Operation Volumes}\label{sec:heuristics}

\newcommand{\trans}{\ensuremath{\hat{\pi}}\xspace}
In \secref{protocol}, we show that the protocol ensures safety and liveness.
However, the proof assumes that the air vehicle does not violate its \contract.
In this section, we discuss how to use existing reachability analyses to over-approximate regions of space-time an air vehicle may visit.
This over-approximation can be used to
\begin{inparaenum}[(i)]
\item generate \Contracts that are unlikely to be violated, or
\item monitor air vehicles at runtime to predict and avoid possible violations.
\end{inparaenum}

Formally, given a dynamical system with state space $D$,
a set of initial states $Q_0 \subseteq D$, and a time horizon $[T_0, T_1)$,
reachability analysis tools can compute \emph{reachtube}, a set of states $Q_1$ reachable within $[T_0, T_1)$.
We further require a function $\trans: \pow{D} \mapsto \pow{\X}$ to transform state space to air-space.
Then, one can build an \contract
\(C_{reach} = (-\infty, \trans(Q_0)), (T_0, \trans(Q_1)), (T_1, \X)\).
This represents that, when air vehicle stays within $\trans(Q_0)$ before $T_0$,
it will stay within $\trans(Q_1)$ between $T_0$ and $T_1$,
and it can be anywhere after $T_1$.
We then can merge $C_{reach}$ for different time horizons to propose \Contracts.

In this work, we use DryVR~\cite{FanQMVCAV017} to compute reachtubes from simulation traces.
DryVR uses collected traces to learn the sensitivity of the trajectories of the air vehicle,
and generates reachtubes for a new simulation trace with probabilistic guarantees.
We use DryVR to study waypoint following for a quadcoptor model, Hector Quadrotor~\cite{hector_quadrotor}
and a fixed-wing model, ROSplane~\cite{ROSplane}, using the Gazebo simulator.

\paragraph{Hector Quadrotor}

\begin{figure}[tb!]
    \includegraphics[width=0.5\textwidth]{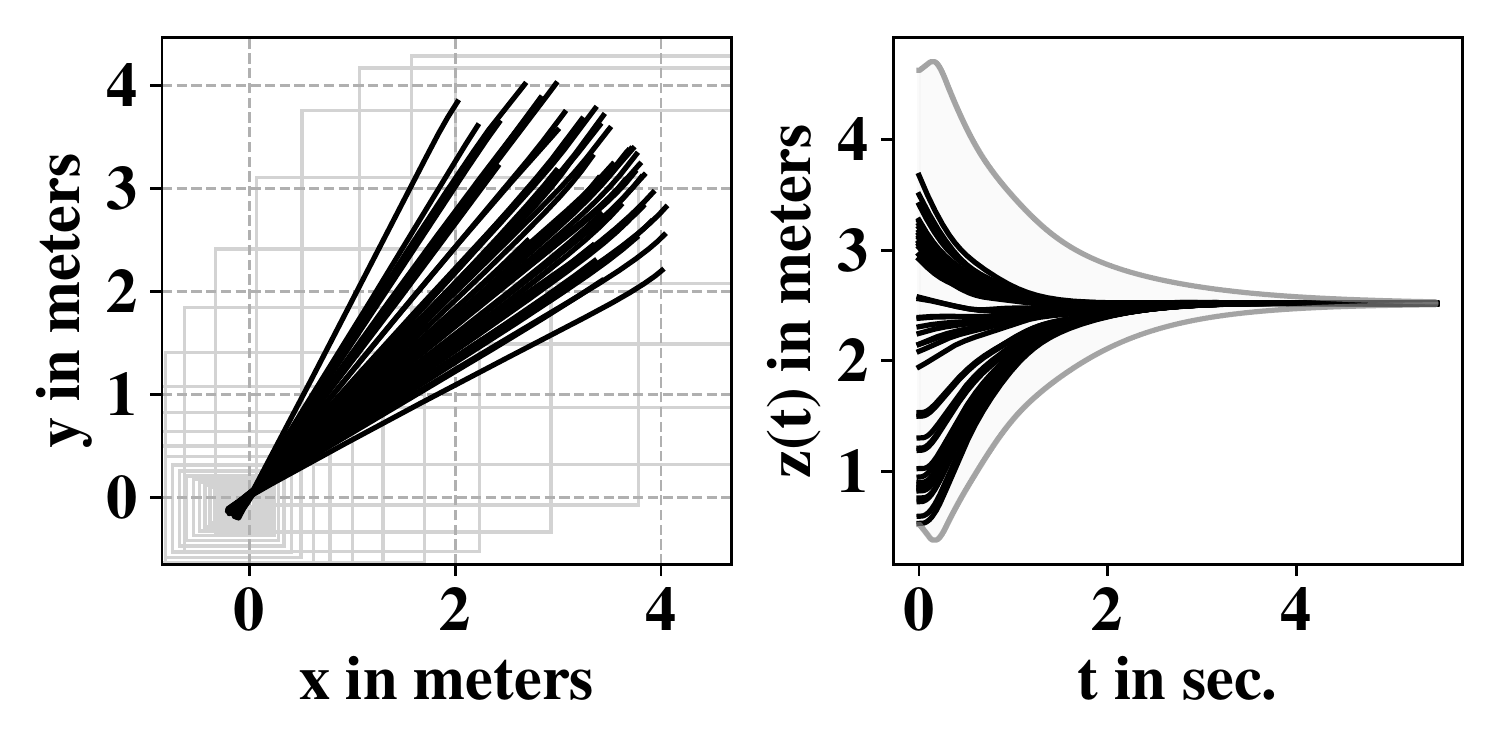}
    \caption{Simulation traces in \emph{Black} and boundary of the reachtube computed by DryVR in \emph{Gray}
             for Hector Quadrotor going to the waypoint at (0, 0, 2.5).
             The reachtube is projected to xy-plane (\emph{Left}) and z-axis over time (\emph{Right}).}
    \label{fig:hector-quad}
\end{figure}

The state variables for Hector Quadrotor already include $x$, $y$, and $z$  for positions but also other variable for orientation and velocity.
Hence, \trans for this model is to simply apply projections to $x$, $y$, and $z$ axes.
We compute $C_{reach}$ for a scenario which the air vehicle goes to the waypoint (0, 0, 2.5).
\figref{hector-quad} shows the projection of $C_{reach}$ as hyper-rectangles to xy-plane~(left) and to z-axis against time~(right).
Observe the projection to xy-plane in \figref{hector-quad},
we can generate \Contracts using a \Conse strategy that covers $C_{reach}$ for the entire time horizon with a bounding rectangle,
or an \Aggre strategy to use the gray rectangles as an \contract with short time intervals.
In general, we can generate a spectrum of \Contracts from $C_{reach}$ between \Conse and \Aggre strategies,
and all \Contracts in this spectrum can guarantee low probability of violations provided by the reachability analysis.
We further explore the performance trade-off between strategies in \secref{experiment}.

\paragraph{ROSplane}

Similarly, the state variables for ROSplane also include $x$, $y$, and $z$ for positions but in North-East-Down~(NED) coordinates.
Hence, \trans for this model is to apply projections to $x$, $y$, and $z$ axes and transform to the coordinates used by the Airspace Manager.
We collect the traces and then divide traces into segments to analyze several path primitives denoted as modes for ROSplane~\cite{ROSplane}.
In \figref{rosplane}, we show the reachtubes for two modes, namely loiter and descend.
Unsurprisingly, the plane may not maintain the desired altitude (z-axis) precisely while loitering,
and thus it is important to reserve enough range of altitude in \Contracts for ROSplane.

\begin{figure}[t!]
    \begin{subfigure}[c]{0.4\columnwidth}
        \includegraphics[height=5cm, clip]{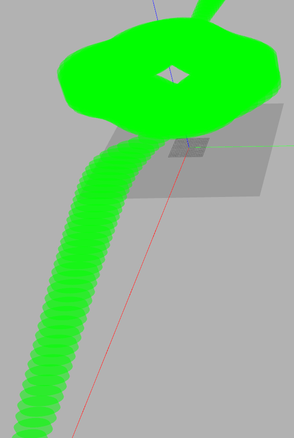}
    \end{subfigure}
    \begin{subfigure}[c]{0.5\columnwidth}
        \begin{subfigure}[t]{0.5\columnwidth}
            \includegraphics[height=2.5cm, clip]{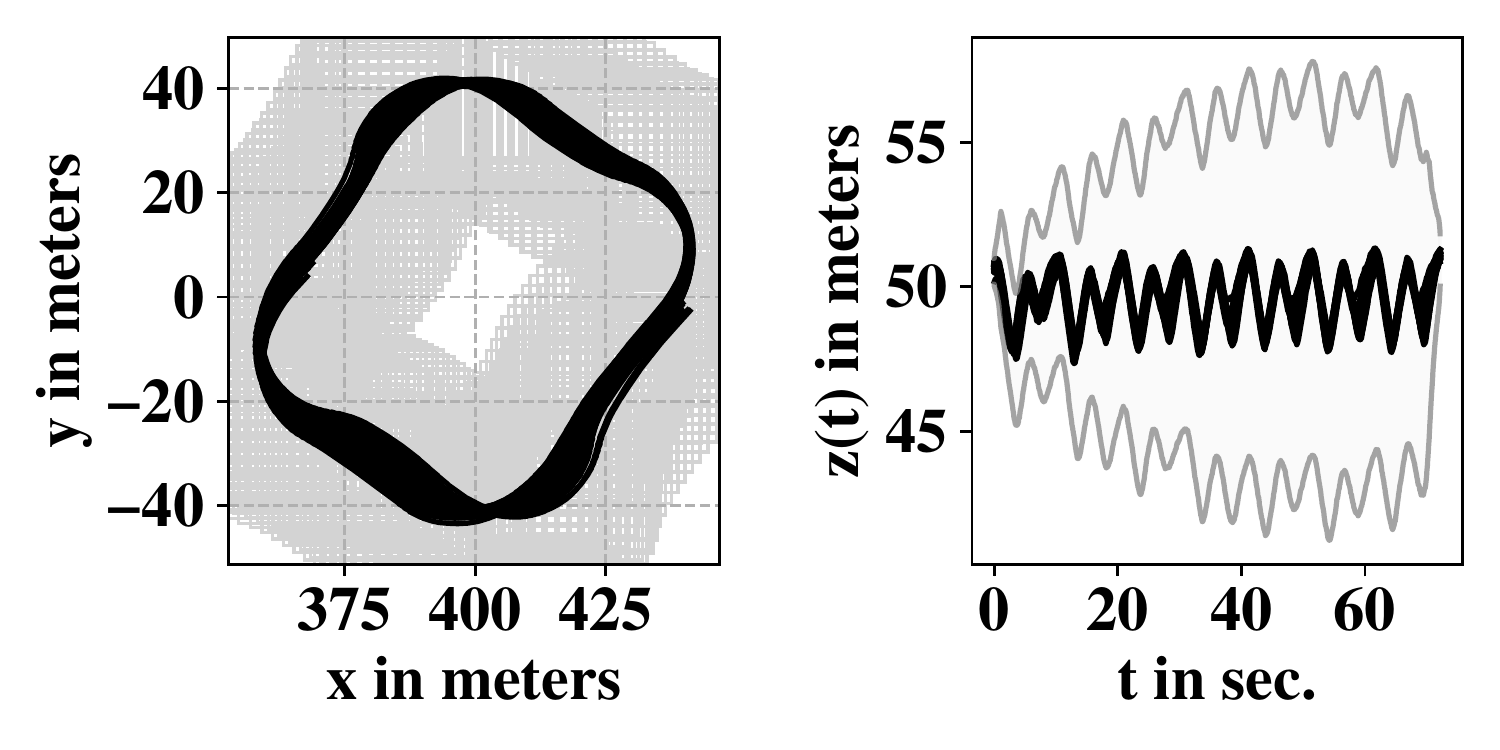}
        \end{subfigure}

        \begin{subfigure}[t]{0.5\columnwidth}
            \includegraphics[height=2.5cm, clip]{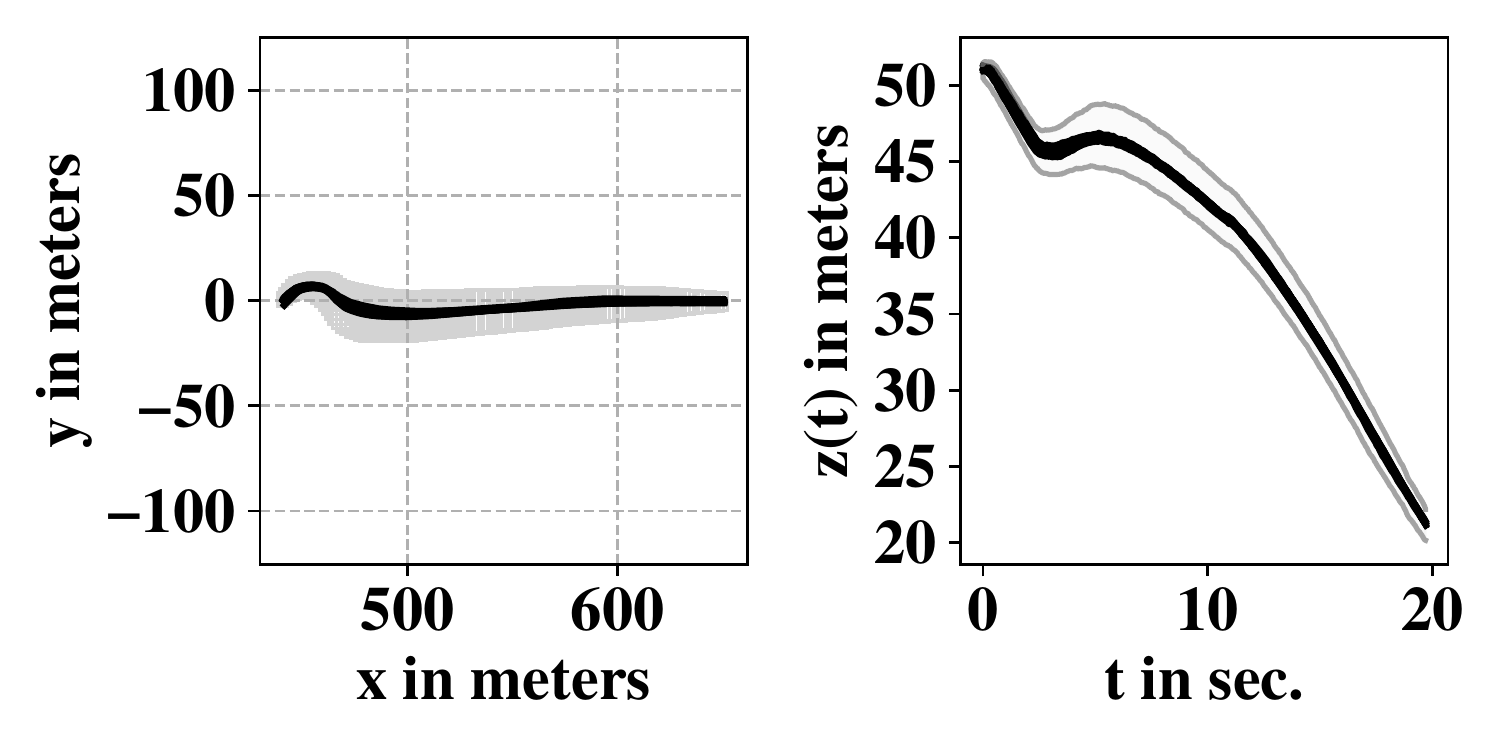}
        \end{subfigure}
    \end{subfigure}

    \caption{Reachtube by DryVR in 3D (\emph{Left}) for ROSplane to loiter and then descend.
        The traces and reachtube for loiter~(\emph{Top Row}) and descent~(\emph{Bottom Row}) are projected to xy-plane (\emph{$1^{\mathit{st}}$ column}) and z-axis over time (\emph{$2^{\mathit{nd}}$ column}).}
    \label{fig:rosplane}
\end{figure}

In summary, we are able to derive useful, i.e., not overly conservative, \opervols using reachtubes from DryVR even with simulations with noises as shown in \figref{rosplane}.
The main engineering difficulty we faced in using DryVR is to divide traces into proper segments that are from the same mode for ROSplane.
This requires domain knowledge on each air vehicle model,
and we refer readers to \cite{hector_quadrotor} and \cite{ROSplane}.

\section{\utmtool implementation and evaluation}
\label{sec:experiment}

\newcommand{\Corr}{\textsc{Corridor}\xspace}
\newcommand{\Loop}{\textsc{Loop}\xspace}
\newcommand{\RND}[1]{\textsc{Random#1}\xspace}
\newcommand{\City}{\textsc{CitySim}\xspace}

\newcommand{\Arena}{\ensuremath{25m\times 25m} arena\xspace}
\newcommand{\Repeat}{three\xspace}

\newcommand{\Agt}{\#\textbf{A}\xspace}
\newcommand{\Qe}{\#\textbf{Q\textsubscript{e}}\xspace}
\newcommand{\Rect}{\#$Rect$\xspace}
\newcommand{\Vio}{\textbf{\%V}\xspace}

Our experiment is conducted using \utmtool.
\utmtool and all simulation scripts are available at our GitHub repository.\footnote{\url{https://github.com/cyphyhouse/CyPhyHouseExperiments}}
To better present our result within page limits,
we only include experiments with the Hector Quadrotor model~\cite{hector_quadrotor} with its default waypoint following controller.
We first describe \utmtool, then the scenarios, and experiment results followed by a brief discussion.

\subsection{\utmtool: System Details}
\label{subsec:toolkit}

\utmtool consists of four major components:
\begin{inparaenum}[(1)]
    \item Dione verification discussed in~\secref{protocol},
    \item reachability analysis and reachtubes from DryVR described in~\secref{heuristics},
    \item an executable reference UTM protocol implemented in Python of Section~\ref{sec:protocol}, and
    \item UTM protocol simulation and visualization with CyPhyHouse~\cite{cyphyhouse_icra2020}.
\end{inparaenum}
Here we focus on the executable UTM protocol and simulation.

To faithfully follow the semantics of our example UTM protocol,
we first provide a data structure to represent and easily manipulate rectangular \opervols.
We provide APIs for designing executable (timed) input/output automata that can interact with simulated vehicles in CyPhyHouse,
and implement an execution engine to simulate the input/output automata alongside CyPhyHouse.
To reuse reachtube from DryVR, we also design APIs to load pre-computed reachtubes for estimating \opervols.
Finally, we also provide several scripts to setup desired scenarios and environments in CyPhyHouse,
and implement a plugin to better visualize \opervols in the Gazebo simulation backend of CyPhyHouse.

\subsection{Evaluation Scenarios}
Following the protocol defined in~\secref{protocol},
a \emph{scenario} for evaluation is specified by
\begin{inparaenum}[(1)]
\item the set of agents \ID\ which we consider $\Agt=|\ID|$
\item the world map and the predefined sequence of waypoints for each agent denoted as the \emph{map}, and
\item the strategy  agents use to to generate \Contracts from their waypoints.
\end{inparaenum}
For example, the \emph{Left} figure in~\figref{maps} shows a scenario with $\Agt=6$ drones in the \Corr map.
It uses \Aggre strategy to generate \Contracts visualized as the red and blue frames.

We evaluate our protocol in the following maps shown in~\figref{maps}:
\begin{enumerate}[(1)]
    \item \Corr simulates two sets of drones on the opposite sides of a tight air corridor trying to pass through.
          This may happen in a garage-like space where a fleet of air vehicles enter or leave.
    \item \Loop simulates each drone following the vertices of the same closed polygonal chain.
          This models common segments in the routes for all air vehicles such as pickup packages or return to base.
    \item \City is a more realistic scenario which simulates drones flying in a city block.
    \item \RND{$N$} are scenarios where each drone follows a sequence of $N$ random waypoints inside a \Arena.
          This is to validate our protocol via random testing.
\end{enumerate}
In addition, a designated landing spot for each drone is specified as the last waypoint in all maps to ensure the liveness property.
This avoids the situation where a landed drone blocks other air vehicles.

\begin{figure*}[t!]
    \centering
    \hfill
    \includegraphics[width=0.3\textwidth]{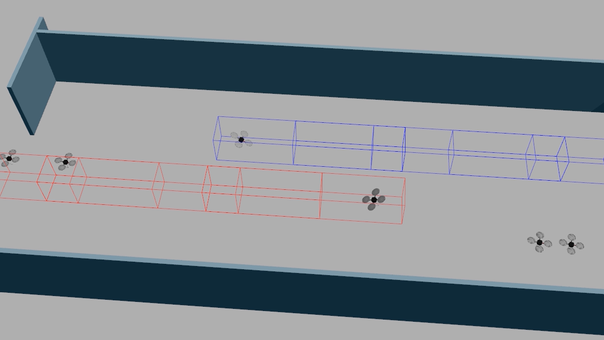}
    \hfill
    \includegraphics[width=0.3\textwidth]{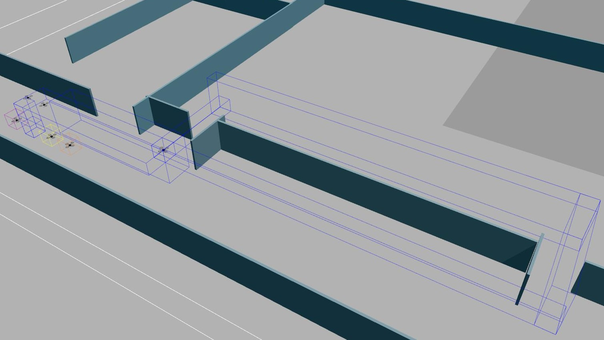}
    \hfill
    \includegraphics[width=0.3\textwidth]{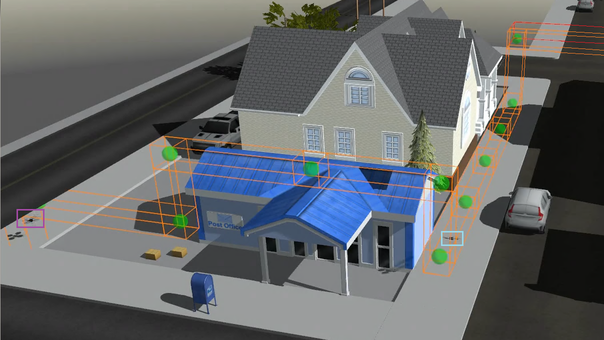}
    \hfill
    \caption{Maps: \Corr~(\emph{Left}), \Loop~(\emph{Mid}), and \City~(\emph{Right})}\label{fig:maps}
\end{figure*}

\paragraph*{\Conse and \Aggre \opervols}
We implement two strategies, namely \Conse and \Aggre, for generating \Contracts from given waypoints and positions.
Both strategies are deterministic and use only \emph{hyper-rectangles} for specifying regions in \Contracts.
As discussed in \secref{heuristics},
\Conse reserves large rectangles covering consecutive waypoints with longer durations between time points.
Thus, it acquires unnecessarily large volumes and may obstruct other agents.
In contrast, \Aggre heuristically selects smaller rectangles and shorter durations.
Therefore, \Aggre is less likely to block other agents but increases the workload of the \AirMgr
because the \opervols (numbers of rectangles) are more complex.

\subsection{Experimental Results}
\paragraph*{Setup}
Our simulation experiments were conducted on a machine with 4 CPUs at 3.40GHz, 8GB memory, and a NVidia GeForce GTX 1060 3GB video card.
The software platform is Ubuntu 16.04 LTS with ROS Kinetic and Gazebo~9.
For the time usage, we report the simulation time from Gazebo (time elapsed in the simulated world), instead of wall clock time
to help reduce the variations in the results due to irrelevant workload on our machine.
To address the nondeterminism arising from concurrency in simulating multiple agents,
we simulate each scenario \Repeat times,
and report the average value of each metric.

\begin{figure*}[t!]

\pgfplotstableread[col sep = comma]{plot-csv/simtime.csv}\SimTime

\begin{tikzpicture}
\begin{axis}[
width=0.38\textwidth,
height=4.5cm,
xlabel={\Agt},
x label style={at={(current axis.south)},anchor=center},
ylabel={Time (sec)},
y label style={at={(current axis.north)},rotate=270,anchor=south},
xmin=1.5, xmax=10.5,
ymin=0, ymax=500,
xtick={2,4,6,8,10},
ytick={0, 250, 500},
legend pos=north west,
ymajorgrids=true,
grid style=dashed,
legend style={font=\tiny, at={(0, 1)}},
]

\addplot [mark=*, line width=1pt]
    table [x=A, y=CORR-MAX] {\SimTime};
\addplot [mark=diamond*, line width=1pt]
    table [x=A, y=LOOP-MAX] {\SimTime};
\addplot [mark=square*, line width=1pt]
    table [x=A, y=RND4-MAX] {\SimTime};
\addplot [mark=triangle*, line width=1pt]
    table [x=A, y=RND6-MAX] {\SimTime};

\addplot [mark=o, dotted, mark options={solid}]
    table [x=A, y=CORR-AVG] {\SimTime};
\addplot [mark=diamond,dotted, mark options={solid}]
    table [x=A, y=LOOP-AVG] {\SimTime};
\addplot [mark=square, dotted, mark options={solid}]
    table [x=A, y=RND4-AVG] {\SimTime};
\addplot [mark=triangle, dotted, mark options={solid}]
    table [x=A, y=RND6-AVG] {\SimTime};

\legend{\Corr, \Loop, \RND{4}, \RND{6}}

\end{axis}
\end{tikzpicture}
\pgfplotstableread[col sep = comma]{plot-csv/avg-query.csv}\AvgQe
\begin{tikzpicture}
\begin{axis}[
width=0.38\textwidth,
height=4.5cm,
xlabel={\Agt},
x label style={at={(current axis.south)},anchor=center},
ylabel={\Qe/s},
y label style={at={(current axis.north)},rotate=270,anchor=south},
xmin=1.5, xmax=10.5,
ymin=0, ymax=25,
xtick={2,4,6,8,10},
ytick={0, 10, 20},
ymajorgrids=true,
grid style=dashed,
]

\addplot [mark=o]
table [x=A, y=CORR-QE] {\AvgQe};
\addplot [mark=diamond]
table [x=A, y=LOOP-QE] {\AvgQe};
\addplot [mark=square]
table [x=A, y=RND4-QE] {\AvgQe};
\addplot [mark=triangle]
table [x=A, y=RND6-QE] {\AvgQe};

\end{axis}
\end{tikzpicture}
\begin{tikzpicture}
\begin{axis}[
width=0.38\textwidth,
height=4.5cm,
xlabel={\Agt},
x label style={at={(current axis.south)},anchor=center},
ylabel={\Rect/s},
y label style={at={(current axis.north)},rotate=270,anchor=south},
xmin=1.5, xmax=10.5,
ymin=0, ymax=25,
xtick={2,4,6,8,10},
ytick={0, 10, 20},
ymajorgrids=true,
grid style=dashed,
]

\addplot [mark=o]
table [x=A, y=CORR-REG] {\AvgQe};
\addplot [mark=diamond]
table [x=A, y=LOOP-REG] {\AvgQe};
\addplot [mark=square]
table [x=A, y=RND4-REG] {\AvgQe};
\addplot [mark=triangle]
table [x=A, y=RND6-REG] {\AvgQe};

\end{axis}
\end{tikzpicture}

\caption{Response time per agent~(\emph{Left}),
    \#emptiness queries per second~(\emph{Mid}),
    and \#rectangles checked by the \AirMgr per second~(\emph{Right}) for each map using \Conse strategy.
    Max is in \emph{Solid marks and lines} and Avg. is in \emph{Hollow marks and dotted lines}}\label{fig:sim-time}\label{fig:queries}
\end{figure*}
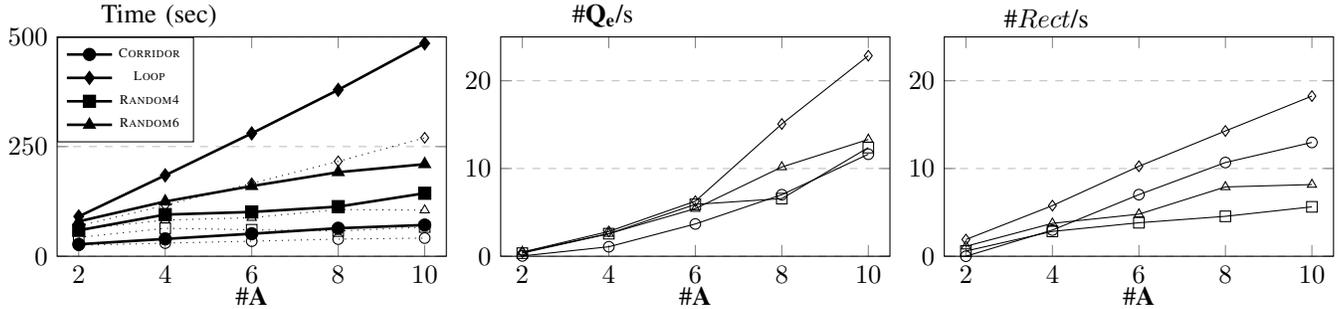

\paragraph*{Response Time and Workload}
\figref{sim-time} shows the response time for each drone starting from sending the first request to finish traversing all waypoints using \Conse strategy in \Corr, \Loop, and \RND{$N$} maps.
As expected, the maximum response time per agent grows linearly against the number of participating agents because, in the worst case,
all agents are accessing the shared narrow air-corridor, and the last agent has to wait until all other agents to finish.
The average response time shows that it is possible to finish faster if agents can execute concurrently in disjoint airspaces.
For example, the average time for 10 agents is smaller the time for 8 agents in \RND{6}.

In \figref{queries},
we consider the number of emptiness/disjointness queries~(denoted as \Qe) and
of hyper-rectangles to check~(denoted as \Rect) per second for the \AirMgr.
\Rect provides a finer estimation of computation resources needed by the \AirMgr than \Qe.
The growth of \Qe as expected is roughly quadratic against \Agt in the worst scenario
due to checking pairwise disjointness.
However, the growth of \Rect is not as fast and seemingly linear to \Agt in the worst scenario.
Therefore, it is very likely the workload increases only linearly instead of quadratically
when we use simple representation of \Contracts such as hyper-rectangles.

\newcolumntype{H}{>{\setbox0=\hbox\bgroup}c<{\egroup}@{}}

\begin{table}[t!]
\begin{center}
    \caption{Comparison of simulation time between \Conse and \Aggre.
    \Agt is the number of agents,
    Time(s) is the total time for simulation according to \emph{the simulated clock} in seconds,
    \Rect/s is the number of rectangles per second in the disjointness query of \Contracts by the \AirMgr.
}\label{table:comparison}%
    \setlength{\tabcolsep}{3pt}
    \scriptsize
        \begin{tabular}{|l|r|r|rH|r|rH|r|r|}
            \hline
            &      & \multicolumn{3}{c|}{\Conse} & \multicolumn{3}{c|}{\Aggre} &         & Increased \\ \cline{3-8}
 Map        & \Agt & Time(s) & \Rect/s &    \Vio & Time(s) & \Rect/s &    \Vio & Speedup &   \Rect/s \\ \hline
            &    2 &   27.52 &    0.00 &     0\% &   21.30 &    0.00 &     0\% &   1.29X &       N/A \\
            &    4 &   39.78 &    2.99 &     0\% &   27.24 &    6.16 &  3.57\% &   1.46X &     2.71X \\
 \Corr      &    6 &   51.63 &    7.02 &     0\% &   34.14 &   14.10 &  2.38\% &   1.51X &     2.06X \\
            &    8 &   64.18 &   10.68 &     0\% &   37.91 &   22.13 &  2.80\% &   1.69X &     2.01X \\
            &   10 &   95.47 &   12.97 &     0\% &   41.94 &   35.14 &  2.24\% &   2.28X &     2.07X \\ \hline
            &    2 &   91.05 &    1.91 &     0\% &   37.63 &    6.85 &  4.53\% &   2.42X &     3.59X \\
            &    4 &  184.88 &    5.77 &  2.08\% &   70.89 &   23.33 &  1.77\% &   2.61X &     4.04X \\
 \Loop      &    6 &  280.51 &   10.26 &  5.56\% &  103.28 &   40.52 &  7.34\% &   2.72X &     3.95X \\
            &    8 &  379.53 &   14.28 &  7.29\% &  134.62 &   63.71 &  8.01\% &   2.82X &     4.46X \\
            &   10 &  485.58 &   18.26 &  5.83\% &  169.25 &   90.94 &  9.48\% &   2.87X &     4.98X \\ \hline
 \City      &    2 &   77.42 &    1.77 &     0\% &   49.92 &    4.48 &  1.19\% &   1.55X &     2.53X \\ \hline
        \end{tabular}
    \end{center}
\end{table}

\paragraph*{\Conse vs. \Aggre.}
We compare the time between \Conse and \Aggre strategies in the \Corr, \Loop, and \City.
Due to the heavier demand for computational resources required, we only simulated with two drones for \City.
\tabref{comparison} shows that \Aggre strategy can reduce the overall response time and provide 1.3-2.8X speedup with larger number of participating agents.
This experiment shows that our framework is suitable for comparing and quantifying the trade-off between performance, safety, and workload under different \Contracts generation strategies.

\section{Discussions and Conclusions}
\label{sec:conclusion}
There is a strong need for a toolkit for formal safety analysis and larger scale empirical evaluations of different UTM concepts and protocols.
In this paper, we present \utmtool, a toolkit with an executable formal model of UTM operations and study its safety, scalability, and performance.

Our toolkit \utmtool offers open and flexible reference implementation of a UTM coordination protocol using ROS and Gazebo.
Our formal analyses in \utmtool illustrate how formal reasoning can be applied to the family of UTM de-conflicting protocols.
We discovered the capability but also the lack of features of Dione~\cite{dione2019} and Dafny~\cite{leino_dafny:_2010} for providing automated proofs,
and to our knowledge, there is no other proof assistant for IOA that also supports the modeling of \Contracts.
We further studied the connection between \Contracts and the reachabilty analysis,
and we showcased how to use DryVR to over-approximate the reachable regions of airspace using simulation traces.
The simulator also makes it possible to study different strategies for reserving \Contracts.

Some of the simplifying assumptions made and consequently limitations of \utmtool can be removed with careful engineering, while others require brand new ideas.
Handling timing and positioning inaccuracies, heterogeneous vehicles, fall in the former category.
We have partly addressed this category using existing reachability analyses in \secref{heuristics}.
In the latter category, a major concern is when there are unavoidable violations of \Contracts due to, for example, hardware failures.
Integration with existing predictive failure detection or failure mitigation strategies and collision avoidance protocols,
incorporation of human operators, or notifications to other participating agents for collision avoidance.
Finally, an important extension is the coordination protocol for multiple airspace managers that still ensures liveness and safety.

\section*{Acknowledgment}
The authors were supported in part by research grants from the National Science Foundation (CyPhyHouse: 1629949 and FMitF: 1918531) and The Boeing Company.
We thank John L. Olson, Aaron A. Mayne, and Michael R. Abraham from The Boeing Company for valuable technical discussions.

\bibliographystyle{IEEEtran}
\bibliography{sayan1,hussein}

\begin{thebibliography}{10}
\providecommand{\url}[1]{#1}
\csname url@samestyle\endcsname
\providecommand{\newblock}{\relax}
\providecommand{\bibinfo}[2]{#2}
\providecommand{\BIBentrySTDinterwordspacing}{\spaceskip=0pt\relax}
\providecommand{\BIBentryALTinterwordstretchfactor}{4}
\providecommand{\BIBentryALTinterwordspacing}{\spaceskip=\fontdimen2\font plus
\BIBentryALTinterwordstretchfactor\fontdimen3\font minus
  \fontdimen4\font\relax}
\providecommand{\BIBforeignlanguage}[2]{{%
\expandafter\ifx\csname l@#1\endcsname\relax
\typeout{** WARNING: IEEEtran.bst: No hyphenation pattern has been}%
\typeout{** loaded for the language `#1'. Using the pattern for}%
\typeout{** the default language instead.}%
\else
\language=\csname l@#1\endcsname
\fi
#2}}
\providecommand{\BIBdecl}{\relax}
\BIBdecl

\bibitem{faa_forecast}
\BIBentryALTinterwordspacing
{Federal Aviation Administration}. (2020) {FAA} {Aerospace} {Forecast} {Fiscal}
  {Year} 2020-2040. [Online]. Available:
  \url{https://www.faa.gov/data_research/aviation/aerospace_forecasts/media/FY2020-40_FAA_Aerospace_Forecast.pdf}
\BIBentrySTDinterwordspacing

\bibitem{utm_conops}
\BIBentryALTinterwordspacing
------. (2020, Mar.) Unmanned {Aircraft} {System} {Traffic} {Management}
  ({UTM}) {Concept} of {Operations} {Version} 2.0. [Online]. Available:
  \url{https://www.faa.gov/uas/research_development/traffic_management/media/UTM_ConOps_v2.pdf}
\BIBentrySTDinterwordspacing

\bibitem{upp_summary}
\BIBentryALTinterwordspacing
------. (2019, Oct.) {UTM} {Pilot} {Program} ({UPP}) {Summary} {Report}.
  [Online]. Available:
  \url{https://www.faa.gov/uas/research_development/traffic_management/utm_pilot_program/media/UPP_Technical_Summary_Report_Final.pdf}
\BIBentrySTDinterwordspacing

\bibitem{cyphyhouse_icra2020}
R.~Ghosh, J.~P. Jansch-Porto, C.~Hsieh, A.~Gosse, M.~Jiang, H.~Taylor, P.~Du,
  S.~Mitra, and G.~Dullerud, ``Cyphyhouse: A programming, simulation, and
  deployment toolchain for heterogeneous distributed coordination,'' 10 2019.

\bibitem{dione2019}
C.~Hsieh and S.~Mitra, ``Dione: A protocol verification system built with dafny
  for i/o automata,'' in \emph{Integrated Formal Methods}, W.~Ahrendt and S.~L.
  Tapia~Tarifa, Eds.\hskip 1em plus 0.5em minus 0.4em\relax Cham: Springer
  International Publishing, 2019, pp. 227--245.

\bibitem{leino_dafny:_2010}
K.~R.~M. Leino, ``\BIBforeignlanguage{en}{Dafny: {An} {Automatic} {Program}
  {Verifier} for {Functional} {Correctness}},'' in
  \emph{\BIBforeignlanguage{en}{LPAR'10}}, ser. LNCS.\hskip 1em plus 0.5em
  minus 0.4em\relax Springer Berlin Heidelberg, 2010, pp. 348--370.

\bibitem{FanQMVCAV017}
C.~Fan, B.~Qi, S.~Mitra, and M.~Viswanathan, ``Dryvr: Data-driven verification
  and compositional reasoning for automotive systems,'' in \emph{Proceedings of
  the 29th International Conference on Computer Aided Verification ({CAV}
  2017)}, 2017, pp. 441--461.

\bibitem{hector_quadrotor}
J.~Meyer, A.~Sendobry, S.~Kohlbrecher, U.~Klingauf, and O.~von Stryk,
  ``Comprehensive simulation of quadrotor uavs using ros and gazebo,'' in
  \emph{Simulation, Modeling, and Programming for Autonomous Robots}, I.~Noda,
  N.~Ando, D.~Brugali, and J.~J. Kuffner, Eds.\hskip 1em plus 0.5em minus
  0.4em\relax Berlin, Heidelberg: Springer Berlin Heidelberg, 2012, pp.
  400--411.

\bibitem{ROSplane}
G.~{Ellingson} and T.~{McLain}, ``Rosplane: Fixed-wing autopilot for education
  and research,'' in \emph{2017 International Conference on Unmanned Aircraft
  Systems (ICUAS)}, 2017, pp. 1503--1507.

\bibitem{tcas_ii}
\BIBentryALTinterwordspacing
{Federal Aviation Administration}. (2011, Feb.) Introduction to {TCAS} {II}
  version 7.1. [Online]. Available:
  \url{https://www.faa.gov/documentlibrary/media/advisory_circular/tcas%20ii%20v7.1%20intro%20booklet.pdf}
\BIBentrySTDinterwordspacing

\bibitem{kochenderfer2012next}
M.~J. Kochenderfer, J.~E. Holland, and J.~P. Chryssanthacopoulos,
  ``Next-generation airborne collision avoidance system,'' Massachusetts
  Institute of Technology Lincoln Laboratory, Tech. Rep., 2012.

\bibitem{ACAS}
M.~J. {Kochenderfer}, C.~{Amato}, G.~{Chowdhary}, J.~P. {How}, H.~J.~D.
  {Reynolds}, J.~R. {Thornton}, P.~A. {Torres-Carrasquillo}, N.~K. {Ure}, and
  J.~{Vian}, \emph{Optimized Airborne Collision Avoidance}, 2015, pp. 249--276.

\bibitem{ACASXu}
K.~D. {Julian}, J.~{Lopez}, J.~S. {Brush}, M.~P. {Owen}, and M.~J.
  {Kochenderfer}, ``Policy compression for aircraft collision avoidance
  systems,'' in \emph{2016 IEEE/AIAA 35th Digital Avionics Systems Conference
  (DASC)}, 2016, pp. 1--10.

\bibitem{air_safety_protocols_survey}
J.~K. {Kuchar} and L.~C. {Yang}, ``A review of conflict detection and
  resolution modeling methods,'' \emph{IEEE Transactions on Intelligent
  Transportation Systems}, vol.~1, no.~4, pp. 179--189, 2000.

\bibitem{Unmanned_Safety_survey}
X.~Yu and Y.~Zhang, ``Sense and avoid technologies with applications to
  unmanned aircraft systems: Review and prospects,'' \emph{Progress in
  Aerospace Sciences}, vol.~74, pp. 152 -- 166, 2015.

\bibitem{TCAS_verification_1997}
J.~{Lygeros} and N.~{Lynch}, ``On the formal verification of the tcas conflict
  resolution algorithms,'' in \emph{Proceedings of the 36th IEEE Conference on
  Decision and Control}, vol.~2, 1997, pp. 1829--1834.

\bibitem{livadasRTSS99}
C.~Livadas, J.~Lygeros, and N.~A. Lynch, ``High-level modeling and analysis of
  {TCAS},'' in \emph{Proceedings of the 20th IEEE Real-Time Systems Symposium
  (RTSS'99),Phoenix, Arizona}, 1999, pp. 115--125.

\bibitem{AirTrafficManagementSystem1999}
N.~Lynch, ``High-level modeling and analysis of an air-traffic management
  system,'' in \emph{Hybrid Systems: Computation and Control}.\hskip 1em plus
  0.5em minus 0.4em\relax Springer Berlin Heidelberg, 1999, p.~3.

\bibitem{Livadas2000HighlevelMA}
C.~Livadas, J.~Lygeros, and N.~Lynch, ``High-level modeling and analysis of the
  traffic alert and collision avoidance system (tcas),'' \emph{Proceedings of
  the IEEE}, vol.~88, pp. 926--948, 2000.

\bibitem{ACASX_verif_tacas}
J.-B. Jeannin, K.~Ghorbal, Y.~Kouskoulas, R.~Gardner, A.~Schmidt, E.~Zawadzki,
  and A.~Platzer, ``A formally verified hybrid system for the next-generation
  airborne collision avoidance system,'' in \emph{Tools and Algorithms for the
  Construction and Analysis of Systems}.\hskip 1em plus 0.5em minus 0.4em\relax
  Springer Berlin Heidelberg, 2015, pp. 21--36.

\bibitem{ACASX_verif_emsoft}
J.~{Jeannin}, K.~{Ghorbal}, Y.~{Kouskoulas}, R.~{Gardner}, A.~{Schmidt},
  E.~{Zawadzki}, and A.~{Platzer}, ``Formal verification of acas x, an
  industrial airborne collision avoidance system,'' in \emph{2015 International
  Conference on Embedded Software (EMSOFT)}, 2015, pp. 127--136.

\bibitem{reluplex}
G.~Katz, C.~Barrett, D.~L. Dill, K.~Julian, and M.~J. Kochenderfer, ``Reluplex:
  An efficient smt solver for verifying deep neural networks,'' in
  \emph{Computer Aided Verification}, R.~Majumdar and V.~Kun{\v{c}}ak,
  Eds.\hskip 1em plus 0.5em minus 0.4em\relax Cham: Springer International
  Publishing, 2017, pp. 97--117.

\bibitem{JM:2012:small}
T.~Johnson and S.~Mitra, ``A small model theorem for rectangular hybrid
  automata networks,'' 2012.

\bibitem{ALAS2014}
P.~S. Duggirala, L.~Wang, S.~Mitra, C.~Munoz, and M.~Viswanathan, ``Temporal
  precedence checking for switched models and its application to a parallel
  landing protocol,'' in \emph{International Conference on Formal Methods (FM
  2014), Singapore}, 2014.

\bibitem{TaylorDistributedVerification}
H.-D. Tran, L.~V. Nguyen, P.~Musau, W.~Xiang, and T.~T. Johnson,
  ``Decentralized real-time safety verification for distributed cyber-physical
  systems,'' in \emph{Formal Techniques for Distributed Objects, Components,
  and Systems}, J.~A. P{\'e}rez and N.~Yoshida, Eds.\hskip 1em plus 0.5em minus
  0.4em\relax Cham: Springer International Publishing, 2019, pp. 261--277.

\bibitem{formalmethodsDrone}
M.~Webster, M.~Fisher, N.~Cameron, and M.~Jump, ``Formal methods for the
  certification of autonomous unmanned aircraft systems,'' in \emph{Computer
  Safety, Reliability, and Security}, F.~Flammini, S.~Bologna, and
  V.~Vittorini, Eds.\hskip 1em plus 0.5em minus 0.4em\relax Springer Berlin
  Heidelberg, 2011, pp. 228--242.

\bibitem{ModelBasedVerif_SemiAutonomous}
O.~{McAree}, J.~M. {Aitken}, and S.~M. {Veres}, ``A model based design
  framework for safety verification of a semi-autonomous inspection drone,'' in
  \emph{2016 UKACC 11th International Conference on Control (CONTROL)}, 2016,
  pp. 1--6.

\bibitem{UmenoL07}
S.~Umeno and N.~A. Lynch, ``Safety verification of an aircraft landing
  protocol: A refinement approach.'' in \emph{HSCC 2007}, 2007, pp. 557--572.

\bibitem{flybylogic}
Y.~V. Pant, H.~Abbas, R.~A. Quaye, and R.~Mangharam, ``Fly-by-logic: Control of
  multi-drone fleets with temporal logic objectives,'' in \emph{ACM/IEEE
  International Conference on Cyber-Physical Systems (ICCPS)}, 2018.

\bibitem{drona}
A.~Desai, I.~Saha, J.~Yang, S.~Qadeer, and S.~A. Seshia, ``Drona: A framework
  for safe distributed mobile robotics,'' in \emph{Proceedings of the 8th
  International Conference on Cyber-Physical Systems}, ser. ICCPS ’17.\hskip
  1em plus 0.5em minus 0.4em\relax ACM, 2017, p. 239–248.

\bibitem{schouwenaars2006safe}
T.~Schouwenaars, ``Safe trajectory planning of autonomous vehicles,'' Ph.D.
  dissertation, Massachusetts Institute of Technology, 2006.

\bibitem{bharadwaj2019traffic}
S.~Bharadwaj, S.~Carr, N.~Neogi, H.~Poonawala, A.~B. Chueca, and U.~Topcu,
  ``Traffic management for urban air mobility,'' in \emph{NASA Formal Methods
  Symposium}.\hskip 1em plus 0.5em minus 0.4em\relax Springer, 2019, pp.
  71--87.

\bibitem{bharadwaj2021traffic_TCNS}
S.~{Bharadwaj}, S.~P. {Carr}, N.~A. {Neogi}, and U.~{Topcu}, ``Decentralized
  control synthesis for air traffic management in urban air mobility,''
  \emph{IEEE Transactions on Control of Network Systems}, pp. 1--1, 2021.

\bibitem{lynch1996a}
N.~A. Lynch, \emph{Distributed Algorithms}.\hskip 1em plus 0.5em minus
  0.4em\relax Morgan Kaufmann Publishers Inc., 1996.

\end{thebibliography}

\appendices

\section{\opervols are Closed under Set Operations}\label{appx:contract}

\begin{definition}
    \label{def:cont-sat}
    Any \contract $C$ represents a compact subset $\ToSet{C}$ of space-time:
    \begin{align*}
        \ToSet{C} \triangleq & \bigcup_{i=1}^{k-1} \{(r, t) \mid r \in R_i \land T_i \leq t < T_{i+1}\} \\
        & \cup \{(r, t) \mid r \in R_k \land T_k \leq t\}
    \end{align*}
    Further, given the current position $\mathit{pos}$ and clock reading $\mathit{clk}$ of an air vehicle, we say that the
    air vehicle meets the contract $C$ if and only if $(\mathit{pos}, \mathit{clk}) \in \ToSet{C}$.
\end{definition}

\begin{definition}
    Given any \contract $C=(R_1, T_1)$, $\ldots,(R_k,T_k)$,
    ${prepend}(C, T_{pp})$ where $T_{pp} < T_1$,
    $split(C, T_{sp})$ where $T_i<T_{sp}<T_{i+1}$,
    and $append(C, T_{ap})$ where $T_k < T_{ap}$ are defined as,
    \begin{align*}
    prepend(C, T_{pp}) \triangleq & \boldsymbol{(\emptyset,T_{pp})}, (R_1, T_1), \dotsc, (R_k, T_k) \\
    split(C, T_{sp}) \triangleq & (R_1, T_1), \dotsc,(R_i,T_i), \boldsymbol{(R_i,T_{sp})}, \\
    &\hspace{0.5in}(R_{i+1},T_{i+1}),\dotsc, (R_k, T_k) \\
    append(C, T_{ap}) \triangleq &(R_1, T_1), \dotsc, (R_k, T_k), \boldsymbol{(R_k,T_{ap})}
    \end{align*}
    Finally, we define $insert(C, T)$ function over any $T$,
    \[
    insert(C, T) \triangleq\left\{
    \begin{array}{l}
    prepend(C, T) \text{ if } T < T_1 \\
    split(C, T), \text{ if } T_i < T < T_{i+1} \\
    append(C, T), \text{ if } T_k < T \\
    C, \text{otherwise}.
    \end{array}\right.
    \]
\end{definition}%

\begin{lemma}
By definition, the \contract produced by $prepend$, $split$, $append$, and $insert$ functions represents the same set of space-time by $C$.
That is, given any \contract $C$ and time point $T$,
\[
\ToSet{insert(C, T)} = \ToSet{C}
\]
\end{lemma}
With the help of $insert$, we can always align two \contract{}s.
We can then implement intersection, union, and difference on \Contracts on top of the same operators for airspace.

%
%
%
%
%

\begin{proposition}\proplabel{equiv}
Given any \contract $C^a$ and $C^b$,
any set operation $\oplus \in \{\cap, \cup, \setminus\}$,
we have the following equivalences:
\begin{align*}
\ToSet{C^a \oplus C^b} &= \ToSet{C^a} \oplus \ToSet{C^b}. \\
\end{align*}
\end{proposition}
The proof is to expand the definition of $\ToSet{\cdot}$ and skipped here.
Given~\propref{equiv}, \Contracts are closed under all set operations;
hence we drop the $\ToSet{\cdot}$ notation.

\section{Agent Protocol Automaton}\label{appx:protocol}

\begin{figure}[ht!]
    \centering
\begin{lstlisting}[language=NumKoord,breaklines=true]
automaton Agent_i
variables:  // Discrete variables
    status: {IDLE,REQUESTING,WAITING,MOVING,RELEASING} := IDLE
    curr_contr: $\contracts$
    plan_contr: $\contracts$
    free_contr: $\contracts$

input plan_i(contr: $\contracts$) $\label{code:ag-plan}$
    eff:
        if status = IDLE:
           plan_contr := contr $\label{code:ag-wp_to_contr}$ $\label{code:ag-plan_contr}$
           status := REQUESTING $\label{code:ag-requesting}$

output request_i(contr: $\contracts$ = plan_contr) $\label{code:ag-request}$
    pre: status = REQUESTING
    eff: status := WAITING

input reply_i(contr: $\contracts$) $\label{code:ag-reply}$
    eff:
        if status = WAITING:
           if curr_contr $\not\subseteq$ contr:
              warning(``Contract too small'') $\label{code:ag-warning}$
              curr_contr := contr
           if plan_contr $\subseteq$ contr:  $\label{code:ag-contr-approved}$ 
              curr_contr := contr
              status := MOVING
           else:
              free_contr := contr$\setminus$curr_contr $\label{code:ag-free_agent}$ 
              status := RELEASING $\label{code:ag-replan}$ 

variables:  // Continuous variables
    clk: $\nnt$
    pos: $\reals^3$  // Position sensor


output next_region_i():   $\label{code:ag-next_region}$ 
    pre: status = MOVING/\len(plan_contr)>=2
        /\ clk >= plan_contr.$T_2$
    eff: plan_contr.pop_front() $\label{code:ag-remove_passed_region}$ 

internal succeed(): $\label{code:ag-succeed}$ 
    pre: status = MOVING/\len(plan_contr)=1
        /\ (pos, clk) \in plan_contr
    eff:
        free_contr := curr_contr$\setminus$plan_contr $\label{code:ag-succeed-release}$
        status := RELEASING

output violate_i(): $\label{code:ag-fail}$ 
    pre: status=MOVING
        /\ (pos,clk) $\notin$ curr_contr
    eff:
        error(``Current contract is violated'')

output release_i(contr: $\contracts$ = free_contr)
    pre: status = RELEASING
    eff:
        curr_contr := curr_contr $\setminus$ free_contr
        status := IDLE $\label{code:ag-idle}$

\end{lstlisting}
    \caption{Agent automaton}\label{fig:agent}
\end{figure}

The detailed automaton is shown in Figure~\ref{fig:agent}.
The agent protocol has a \inlinekrd{status} variable to keep track of the discrete states in Figure~\ref{fig:stateflow}.
In addition, it uses three contract-typed variables for the following purposes:
\begin{inparaenum}[(1)]
    \item \inlinekrd{curr_contr} is a local copy of the current contract maintained for $i$ by the \AirMgr,
    \item \inlinekrd{plan_contr} is a contract that $i$ wants to propose to the \AirMgr to be able to visit the planned waypoints, and
    \item \inlinekrd{free_contr} tracks the releasable portion of the current contract \inlinekrd{curr_contr}.
\end{inparaenum}
In addition, the agent $i$ can read its
current position from the variable \inlinekrd{pos}  and the current global time from the variable \inlinekrd{clk}.
To provide a simple abstraction of arbitrary controllers for the agent,
we create the variable \inlinekrd{traj_ctrl} that stores a list of waypoints that the agent would follow when it is in the \inlinekrd{MOVING} status.
\inlinekrd{traj_ctrl} has two abstract interfaces: \inlinekrd{set_waypoints} to store the plan waypoints and calculate the necessary control signal (using PID, for example) and \inlinekrd{start} to start moving the agent to follow waypoints.

Each agent $i$ is initialized in \inlinekrd{IDLE} status.
When it receives a \inlinekrd{plan} action with a given \inlinekrd{contr} (line~\ref{code:ag-plan}),
the agent stores \inlinekrd{contr} as \inlinekrd{plan_contr} (line~\ref{code:ag-wp_to_contr}) and enters the \inlinekrd{REQUESTING} status (line~\ref{code:ag-requesting}).
A number of strategies may be followed to create contracts from waypoints lists, for example using reachability analysis  for a given  waypoint-tracking controller for the aircraft, or creating fixed-sized 3D rectangles centered at the segments connecting the waypoints.
We will discuss this further in \secref{experiment}.
Agent $i$ then makes a request \inlinekrd{request_i(contr)} with \inlinekrd{contr=plan_contr}
to denote the planned contract is sent as output,
and enters \inlinekrd{WAITING} status to wait for a reply from the \AirMgr (line~\ref{code:ag-request}).

When agent $i$ receives a \inlinekrd{reply_i(contr)} from the \AirMgr,
the contract \inlinekrd{contr} represents the contract of agent $i$ recorded by the \AirMgr (line~\ref{code:ag-reply}).
It is the union of all contracts agent $i$ have acquired and not yet released.
Agent $i$ first checks whether the contract \inlinekrd{curr_contr} is a subset of \inlinekrd{contr} or not.
If not, it means the local copy is less restrictive, so the \AirMgr may grant contracts to other agents conflicting with agent $i$.
This may lead to a safety violation, and hence agent $i$ raises a warning~(line~\ref{code:ag-reply}).
Otherwise, the agent checks if the contract \inlinekrd{contr} approved by the \AirMgr contains \inlinekrd{plan_contr}, i.e. \inlinekrd{plan_contr} $\subseteq$ \inlinekrd{contr} (line~\ref{code:ag-contr-approved}).
If yes, then it updates its \inlinekrd{curr_contr} to be equal to the new approved \inlinekrd{contr}. The agent then
calls \inlinekrd{traj_ctrl.start} to start following the waypoints,
and transitions to the \inlinekrd{MOVING} status.
If no, i.e. there is a part of \inlinekrd{plan_contr} that is not approved  \inlinekrd{contr} and not approved by the \AirMgr,
then agent $i$ does not change \inlinekrd{curr_contr}. It only checks the part of the contract saved by the \AirMgr that is no longer a part of \inlinekrd{curr_contr} of the agent. It then stores this portion of the contract in \inlinekrd{free_contr} (line~\ref{code:ag-free_agent}), and directly goes to the \inlinekrd{RELEASING} status to release and re-plan (line~\ref{code:ag-replan}).

When the agent is in the \inlinekrd{MOVING} status,
the \inlinekrd{next_region} action will be triggered whenever the global time passes the time bound of a region in the contract (line~\ref{code:ag-next_region}). That action will remove that pair of region and time point from \inlinekrd{plan_contr} (line~\ref{code:ag-remove_passed_region}).
Once there is only a single pair left in the planned contract \inlinekrd{plan_contr} and the contract is not violated,
the \inlinekrd{succeed} action is triggered to indicate the plan is executed successfully (line~\ref{code:ag-succeed}).
Agent $i$ then calculates the releasable contract \inlinekrd{free_contr} to be its contract \inlinekrd{curr_contr} excluding the last pair of \inlinekrd{plan_contr} (line~\ref{code:ag-succeed-release}).
Finally, it enters \inlinekrd{RELEASING} status.
It sends \inlinekrd{release_i(contr)} to notify the \AirMgr the contract that agent $i$ can release,
and goes back to \inlinekrd{IDLE} status (line~\ref{code:ag-idle}).

If at any point in time the current contract is violated, the \inlinekrd{violate} action would be triggered (line~\ref{code:ag-fail}). Remember that the contract is violated if the current pair of position and time of the agent is outside of the space-time specified by the contract.
This can happen in case the agent moves outside a region in a time interval of the contract,
or the agent could not reach a region before its specified time point in the contract. It then declares a violation to the \AirMgr.

\subsection{Safety under Delayed Communication}\label{appx:safety-delayed}

\newcommand{\TsendL}{\ensuremath{T^{\mathtt{rel}}_{snd}}\xspace}
\newcommand{\TrecvL}{\ensuremath{T^{\mathtt{rel}}_{rcv}}\xspace}
\newcommand{\TsendP}{\ensuremath{T^{\mathtt{rep}}_{snd}}\xspace}
\newcommand{\TrecvP}{\ensuremath{T^{\mathtt{rep}}_{rcv}}\xspace}
\newcommand{\TsendQ}{\ensuremath{T^{\mathtt{req}}_{snd}}\xspace}
\newcommand{\TrecvQ}{\ensuremath{T^{\mathtt{req}}_{rcv}}\xspace}

\begin{figure}[t!]
	\centering
	\scriptsize
	\begin{tikzpicture}[scale=0.8]
	\draw (-2,0) -- (-2,-2) (2,0) -- (2,-2);
	\node at (-2, 0.7) {$\agent_i$};
	\node at (-2, 0.3) {\inlinekrd{curr_contr=C_0}};
	\node at ( 2, 0.7) {\AirMgr};
	\node at ( 2, 0.3) {\inlinekrd{contr_arr[i]=C_0}};

	\node [left]  (TsendL) at (-2,-0.5) {\TsendL};
	\node [left]  (TrecvP) at (-2,-1.5) {\TrecvP};
	\node [right] (TsendP) at ( 2,-0.5) {\TsendP};
	\node [right] (TrecvL) at ( 2,-1.5) {\TrecvL};

	\draw[->] (TsendL) -- node[midway, sloped, above, pos=0.25] {\inlinekrd{release_i(C_1)}} (TrecvL);
	\draw[<-] (TrecvP) -- node[midway, sloped, above, pos=0.8] {\inlinekrd{reply_i(C_0)}}    (TsendP);

	\node at (-2.4, -1) {\inlinekrd{curr_contr=C_0\\C_1}};
	\node[text=red] at (-2, -2.3) {\inlinekrd{curr_contr=C_0}};
	\node[text=red] at ( 2, -2.3) {\inlinekrd{contr_arr[i]=C_0\\C_1}};
	\end{tikzpicture}

	\caption{Sequence diagram for an impossible unsafe \contract release.}\label{fig:cex}
\end{figure}
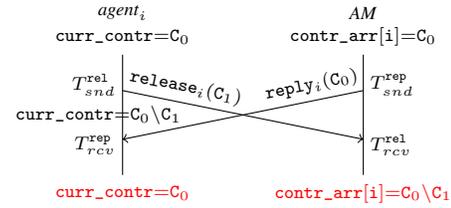

Now we consider the case where actions are delivered with bounded delay. 
Our proof is to show the impossibility of unsafe action sequences under the reliable communication.
Because the current contract of each agent is only updated after receiving \inlinekrd{reply_i} from the \AirMgr
and shrunk when sending \inlinekrd{release_i},
the potential counterexample shown in \figref{cex} can only happen if \inlinekrd{reply_i} is delivered to agent $i$ to update its local copy
while \inlinekrd{release_i} is delivered to the \AirMgr to shrink the global copy concurrently, i.e., $\TsendL<\TrecvP$ and $\TsendP<\TrecvL$.
Recall from \figref{stateflow} that our protocol ensures \inlinekrd{request_i}, \inlinekrd{reply_i}, and \inlinekrd{release_i} happen in such order by design.
We can prove this order of actions by induction on the formally defined automaton but skip the proof here for simplicity.
Therefore, we know that ther must be a \inlinekrd{request_i} sent after \inlinekrd{release_i},
and \inlinekrd{reply_i} is the response to this request.
Now we provide a simplified reliable communication assumption for this proof.
\begin{assumption}\label{ass:comm}
	The reliable communication guarantees the messages sent by the same agent are delivered in order.
	In particular, if $\agent_i$ sends a \inlinekrd{release_i} first and \inlinekrd{request_i} second,
	we denote \TsendL as the time \inlinekrd{release_i} is sent and \TrecvL as the time received,
	similarly \TsendQ and \TrecvQ for \inlinekrd{request_i}. Formally,
	\[
	\TsendL \leq \TsendQ \Rightarrow \TrecvL \leq \TrecvQ
	\]
\end{assumption}
Also by definition, $T^*_{snd} \leq T^*_{rcv}$ because sending must happen before receiving.
The order between actions can be formally specified as $\TsendL \leq \TsendQ \leq \TrecvQ \leq \TsendP$
because the request must have been delivered to the \AirMgr for it to trigger the \inlinekrd{reply_i}.
We can then derive $\TsendL \leq \TsendQ \leq \TrecvQ \leq \TsendP < \TrecvL$.
This contradicts to our assumption of reliable communication because messages from $\agent_i$ are delivered out of order.
To be more precise, \inlinekrd{release_i} is sent before ($\TsendL \leq \TsendQ$)
but delivered later~($\TrecvQ < \TrecvL$) than \inlinekrd{request_i}.
This contradicts to $\TsendL \leq \TsendQ \Rightarrow \TrecvL \leq \TrecvQ$.
Hence, we prove by contradiction.

\section{Airspace Manager Automaton in Dione}\label{appx:airmgr-dione}

The airspace manager automaton~(\AirMgr) modeled in Dione is provide in Figure~\ref{fig:airmgr-dione}.

\begin{figure*}[t!]
    \centering
\begin{lstlisting}[
    language=Python,
    breaklines=true,
    basicstyle=\scriptsize,
    identifierstyle=\ttfamily,
    keywordstyle=\bfseries,
    numberstyle=\tiny,
    stepnumber=1,
    numbersep=4pt,
    numbers=left,
    ]
StampedPoint: type = NamedTuple[t: Real, x: Real, y: Real, z: Real]
Contract: type = ISet[StampedPoint]
UID: type = Nat

@automaton
def AirspaceManager(N: nat):
    where = True

    class signature:
        @input
        def request(i: UID, contr: Contract):
            where = i < N
        @output
        def reply(i: UID, contr: Contract):
            where = i < N
        @input
        def release(i: UID, contr: Contract):
            where = i < N

    class states:
        contr_arr: Seq[Contract]
        reply_set: Set[UID]
    initially = reply_set == set() and len(contr_arr) == N and \
            forall(i, forall(j, implies(0 <= i < N and 0 <= j < N and i != j, disjoint(contr_arr[i], contr_arr[j]))))

    class transitions:
        @input
        def request(i: UID, contr: Contract):
            reply_set = reply_set + {i}
            if forall(j, implies(0 <= j < N and i != j, disjoint(contr, contr_arr[j]))):
                contr_arr[i] = contr_arr[i] + contr

        @output
        @pre(i in reply_set and contr == contr_arr[i])
        def reply(i: UID, contr: Contract):
            reply_set = reply_set - {i}

        @input
        def release(i: UID, contr: Contract):
            contr_arr[i] = contr_arr[i] - contr

    invariant_of = len(contr_arr) == N and \
        forall(i, forall(j, implies(0 <= i < N and 0 <= j < N and i != j, disjoint(contr_arr[i], contr_arr[j]))))
\end{lstlisting}
    \caption{Airspace manager automaton in Dione}\label{fig:airmgr-dione}
\end{figure*}

\end{document}